\pdfoutput=1
\documentclass{article}
\usepackage{pdfpages}
\usepackage{fullpage}
\usepackage{graphicx}
\let\OLDthebibliography\thebibliography
\renewcommand\thebibliography[1]{
  \OLDthebibliography{#1}
  \setlength{\parskip}{0pt}
  \setlength{\itemsep}{0pt plus 0.3ex}
}
\usepackage[boxed,ruled,vlined,linesnumbered]{algorithm2e}
\usepackage{url}
\usepackage{framed}
\usepackage{hyperref}
\newcommand*\wrapletters[1]{\wr@pletters#1\@nil}
\def\wr@pletters#1#2\@nil{#1\allowbreak\if&#2&\else\wr@pletters#2\@nil\fi}
\usepackage{caption}
\usepackage{subcaption}
\usepackage[none]{hyphenat}
\usepackage{times}
\usepackage{amsmath}
\usepackage{theorem}
\usepackage{wrapfig}
\usepackage{amssymb}
\usepackage{xspace}

\newtheorem{theorem}{Theorem}
\newtheorem{lemma}[theorem]{Lemma}

\newtheorem{inv}[theorem]{Invariant}

\newcommand{\qed}{\hfill $\blacksquare$}
\newcommand{\qedClaim}{\hfill \ensuremath{\Box}}
\newenvironment{proof}{\noindent {\bf Proof.}\ }{\qed\par\vskip 0mm\par}

\newcommand{\B}{\vspace*{-\smallskipamount}}
\newcommand{\BB}{\vspace*{-\medskipamount}}
\newcommand{\BBB}{\vspace*{-\bigskipamount}}
\usepackage{setspace}
\usepackage{lineno}
\newcommand{\BigO}[0]{{\cal O}\xspace}
\usepackage{csquotes}
\usepackage{enumitem}
\newcommand{\Keywords}[1]{\par\noindent
{\small{\em Keywords\/}: #1}}
\usepackage{multirow}
\usepackage{hhline}

\begin{document}

\title{On Detecting Termination in Cognitive Radio Networks\footnote{\textbf{Accepted in International Journal of Network Management (Wiley-IJNM).}}\footnote{A preliminary version of this paper has appeared in proceeding of the $17^{th}$ Pacific Rim International Symposium on Dependable Computing 2011 (PRDC)~\cite{DBLP:conf/prdc/SharmaS11}. 
}}
\date{}
\author{Shantanu Sharma\thanks{Department of Computer Science, Ben-Gurion University of the Negev, Israel. Email: \texttt{sharmas@cs.bgu.ac.il}.} \quad Awadhesh Kumar Singh\thanks{Department of Computer Engineering, National Institure of Technology, Kurukshetra, Haryana, India. Email: \texttt{aksinreck@ieee.org}.}}

\maketitle

\begin{abstract}
The cognitive radio networks are an emerging wireless communication and computing paradigm. The cognitive radio nodes execute computations on multiple heterogeneous channels in the absence of licensed users (a.k.a. primary users) of those bands. Termination detection is a fundamental and non-trivial problem in distributed systems. In this paper, we propose a termination detection protocol for multi-hop cognitive radio networks where the cognitive radio nodes are allowed to tune to channels that are not currently occupied by primary users and to move to different locations during the protocol execution. The proposed protocol applies credit distribution and aggregation approach and maintains a new kind of logical structure, called the \emph{virtual tree-like structure}. The \emph{virtual tree-like structure} helps in decreasing the latency involved in announcing termination. Unlike conventional tree structures, the \emph{virtual tree-like structure} does not require a specific node to act as the root node that has to stay involved in the computation until termination announcement; hence, the root node may become idle soon after finishing its computation. Also, the protocol is able to detect the presence of licensed users and announce strong or weak termination, whichever is possible.
~\\

\end{abstract}

\pagenumbering{arabic}
\setcounter{page}{1}

\section{Introduction}
\label{section:introduction}
A vast growth of small and portable devices has culminated into the problem of bandwidth scarcity. Hence, it is becoming difficult to provide seamless connectivity while executing various applications, \textit{e}.\textit{g}., email, web surfing, gaming, and video conferencing (see Exhibit 10 in~\cite{FCCreport1}). It is also noted that currently allocated spectrums have their significant portions underutilized~\cite{url1}. The Cognitive Radio Networks (\textit{CRN}s)~\cite{DBLP:journals/cn/FortunaM09} are a smart solution with a complex network structure to enhance the spectrum utilization.

The termination detection~\cite{dijkstra1980termination,Francez:1980:DT:357084.357087} is a fundamental and non-trivial problem in distributed systems because the processors do not have the complete knowledge of all the other processors in the network, and there is no global clock in the distributed computing environment. A solution to the termination detection problem informs termination of the task being executed in the network.

\subsection{Cognitive radio networks}
\label{subsec:The cognitive radio network}
A cognitive radio network (see~\cite{Akyildiz:2009:CCR:1508337.1508834,1391031,Cesana:2011:RCR:1930071.1930095,Wang:2011:ECR:2288683.2294147,Akyildiz:2006:NGS:1162469.1162470,4481339,5639025}) is a collection of heterogeneous cognitive radio nodes (or processors), called secondary users. The cognitive radio nodes (CRs) have sufficient computing power and power backup to operate on multiple heterogeneous channels (or frequency bands) in the absence of the licensed user(s), termed as primary user(s), of the respective bands. The cognitive radio nodes have the \textit{LEIRA} (learning, efficiency, intelligence, reliability, and adaptively) capability to scan and operate on different channels.

A channel that is not currently occupied by a primary user is called an \textit{available channel}. Any two nodes that are in the transmission range of each other and tuned to a common available channel during an identical time interval, are called \textit{neighboring nodes}. The appearance of a primary user (hereafter, the primary users will be known as PUs) on an available channel is a reason for CRs to switch the channel and to tune to another available channel, because the CRs are not allowed to interrupt the primary users in any case.

The modern networking benchmark, \textit{CRN}, presents many unique challenges in the field of communication as well as computing, such as cognitive capability, reliability, and efficiency. Several challenges of \textit{CRN}s are presented in~\cite{Cesana:2011:RCR:1930071.1930095}. Interested readers may refer to~\cite{wyglinski2009cognitive,Mahmoud-2007-book,Hossain-2009-book,Liotta-2013} for more details on cognitive radio networks.

In this paper, unless otherwise indicated, the words \enquote{cognitive radio node,} \enquote{cognitive radio,} \enquote{node,} and \enquote{processor} have the same meaning, and similarly, the words \enquote{cognitive radio network,} \enquote{network,} and \enquote{system} have been treated as synonyms.

\subsection{Termination detection (in cognitive radio networks)}
\label{subsec:Termination detection in the cognitive radio networks}
Nowadays, a large number of distributed applications --– \textit{e}.\textit{g}., mutual exclusion, leader election, checkpointing, global state detection~\cite{Kshemkalyani:2008:DCP:1374804} --– are executed on portable devices. In general, an application that executes on processors is known as a \textit{normal computation} or an \textit{underlying computation}. A termination detection (TD) protocol~\cite{dijkstra1980termination,Francez:1980:DT:357084.357087} is used to announce termination of the normal computation. The termination declaration of a normal computation, when it has indeed terminated --- in a group of mobile devices that are geographically distributed and tuned on different channels --- is an interesting challenge in cognitive radio networks. Hereafter, we use the word \enquote{computation} that refers to the \enquote{normal computation.}

A node may be in \textit{active} or \textit{passive} state during a computation. The nodes in the active state are called \textit{active nodes}, and the nodes in the passive state are called \textit{passive nodes}. The active nodes execute an assigned computation, and usually, after completion of the computation, they become passive. A passive node can become active on reception of a message from an active node. Hence, it is clear that only active nodes can send messages; however, both the active and passive nodes can receive messages at any time.

Initially, all the nodes are passive in the network. Since only active nodes can send messages, we assume that there exist a passive node that becomes active on reception of a message from outside world, and subsequently it initiates the computation. A computation is said to be terminated if and only if all the nodes are passive and there is no message in-transit. A brief summary about TD protocols can be found in Chapter 7 of~\cite{Kshemkalyani:2008:DCP:1374804} and Chapter 9 of~\cite{ghosh2010distributed}.

Any termination detection protocol can be initiated in two ways, as follows:
\begin{itemize}[noitemsep]
  \item \textit{Delayed initiation.} The TD protocol is triggered by any node, $i$, that has been assigned a computation, and the same node $i$ is responsible for the announcement of termination; such an initiation is known as delayed initiation~\cite{DBLP:journals/jpdc/MittalFVP08}. Here, it is not mandatory that the node $i$ was also the initiator of the computation; refer to Figure~\ref{fig:two_way_of_initation_first}, where node 1 initiates the computation and node 3 initiates the termination detection protocol.
  \item \textit{Concurrent initiation.} In the concurrent initiation, the TD protocol is overlaid on a computation and executes concurrently. Here, the initiator of the computation is also responsible for the announcement of termination; refer to Figure~\ref{fig:two_way_of_initation_second}, where the lower part represents the execution of the computation and the upper part represents the execution of termination detection protocol that is being executed concurrently with the computation; and node 1 is the initiator of both, the computation and the termination detection protocol.
\end{itemize}

Any termination detection protocol should satisfy the following properties:
\begin{itemize}[noitemsep]
  \item \textit{No false termination detection} (\textit{safety}). The termination of a computation is declared only when the computation has indeed terminated (and only a single designated node can announce termination).
  \item \textit{Eventual termination detection} (\textit{liveness}). A single (designated) node announces termination within a finite amount of time.
\end{itemize}

The termination detection in \textit{CRN} is more challenging as compared to the conventional wireless networks because of the following reasons:
\begin{figure}[t]
\centering
    \begin{minipage}[b]{.45\textwidth}
        \centering
          \includegraphics[width=60mm, height=60mm]{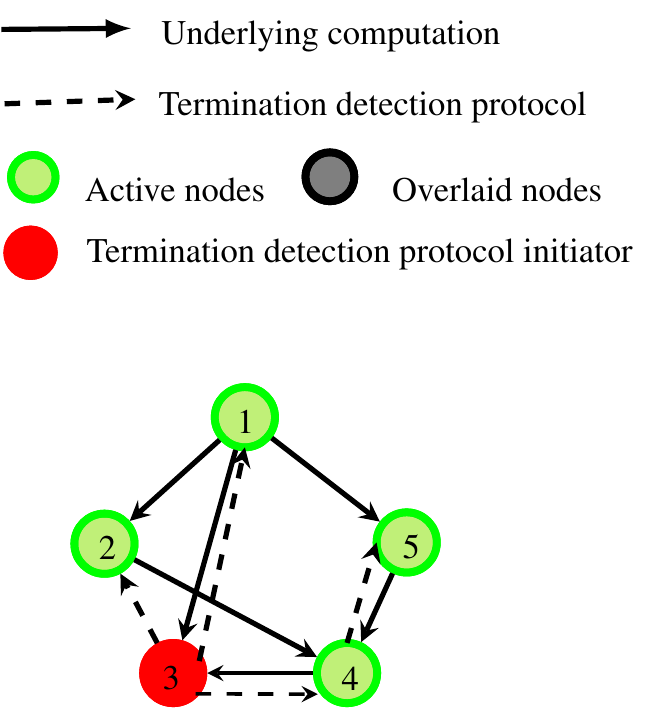}
       \subcaption{Delayed initiation.}
       \label{fig:two_way_of_initation_first}
    \end{minipage}
    \begin{minipage}[b]{.45\textwidth}
        \centering
         \includegraphics[width=60mm, height=60mm]{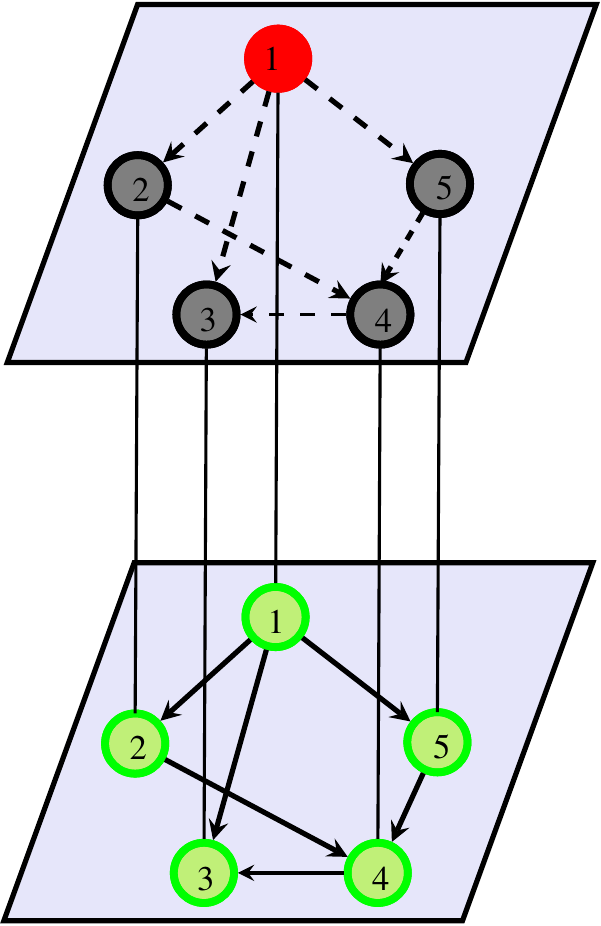}
        \subcaption{Concurrent initiation.}
        \label{fig:two_way_of_initation_second}
    \end{minipage}
\B
\caption{Two ways of termination detection protocol's initiations.}
\BBB
\label{fig:two_way_of_initation}
\end{figure}

\begin{itemize}[noitemsep]
  \item \textit{Network structure and communication links.} The \textit{CRN} is a network of time and space varying channels. Any two neighboring nodes, which must be tuned to an identical available channel, can communicate directly (using a communication link). The appearance of a PU on a channel forces the neighboring CRs to vacate that channel and to tune to another identical available channel. However, finding another identical alternative available channel, for neighboring CRs, is not an easy task due to reasons like topological dynamics and varying capabilities of the nodes~\cite{DBLP:conf/wasa/BansalMV08}. Hence, the communication link endurance during execution of any protocol is hard to guarantee.

  \item \textit{Reaction to a communication link break.} In classical wireless networks, the nodes operate on a pre-decided channel that provides them communication links. The communication links may break due to node mobility or node failure; hence, a new communication link detection takes place in a highly reactive manner without considering parameters like endurance of the link. On the other hand, in \textit{CRN}, the communication links may also break due to the appearance of a PU leading to \textit{spectrum mobility} that emphasizes on several factors before creating the new communication links~\cite{Akyildiz:2006:NGS:1162469.1162470}.

  \item \textit{Sufficient resources.} The computing nodes in wireless domain suffer from limited resources like bandwidth, memory, and battery power. Thus, several protocols focus on the reduction of the number of messages exchanged to minimize the need of bandwidth, memory, and battery power. However, \textit{CRN}s have sufficient resources, especially temporarily unused spectrums (known as \textit{spectrum holes}~\cite{Akyildiz:2006:NGS:1162469.1162470}) and computing power. Consequently, the focus of research has been shifted to other challenges related to the execution of various applications.

  \item \textit{No definitive logical structure.} Most of the computing protocol use quasi-stable logical structures, \textit{e}.\textit{g}., tree, ring, to leverage the design difficulties. The \textit{CRN}s restrict a direct engagement of such logical structures due to time and space varying channels.
\end{itemize}
In addition, unlike other ad hoc networks, the CRs show very loose synchronization, poor tolerance to the heterogeneity of mobile devices as well as channels, and an extra cost for searching a new channel (on the appearance of primary users). The presence of these challenges in cognitive radio networks make the design of computing and communication protocols harder.

\subsection{Our contribution and outline of the paper}
\label{subsec:Our contribution}

The paper presents a concurrent initiation (see Figure~\ref{fig:two_way_of_initation_second}) based Termination detection protocol for Cognitive RAdio Networks, called \textit{T-CRAN}, henceforth. Moreover, our protocol can also be implemented in other dynamic networks, \textit{e}.\textit{g}., cellular networks, mobile ad hoc networks (MANETs), vehicular ad hoc networks (VANETs). In this paper, we provide:

\begin{enumerate}[noitemsep]
  \item A credit distribution and aggregation based termination detection protocol for \textit{CRN}, in Section~\ref{section:The T-CRAN Protocol}, that declares termination of computations despite the presence of PUs. Our protocol recognizes the cognitive radio nodes that lose their single available channel due to the appearance of PUs and are unable to find other available channel.

  \item A new logical structure, called \textit{virtual tree-like structure} (Figure~\ref{fig:The virtual tree-like structure.}), where the root node can be passive when it completes its computation, unlike conventional (logical) tree structures, where it is mandatory for the root node to stay in active state till the end of computation; in Section~\ref{section:The T-CRAN Protocol}.

  \item The \textit{T-CRAN} protocol as guarded-actions, in Section~\ref{section:T-CRAN in Guard-Action Paradigm}. Section~\ref{section:The Working of the T-CRAN} explains the complete working of the proposed protocol.

  \item The analysis of message and time complexities of the proposed protocol, in Appendix~\ref{section:Complexity Analysis}. The correctness proofs of the proposed protocol are given in Appendix~\ref{section:correctness_proof}.
\end{enumerate}

\subsection{Related work}
\label{subsection:Related Work}
The termination detection (TD) protocol has been studied extensively in static distributed systems~\cite{dijkstra1980termination,Francez:1980:DT:357084.357087,DBLP:journals/jpdc/ChandrasekaranV90,DBLP:journals/jise/HuangK91,DBLP:journals/ipl/Mattern89,DBLP:journals/toplas/MisraC82,DBLP:journals/ipl/Topor84,DBLP:conf/gg/GodardMMS02}. A detailed classification of TD protocols is given in~\cite{Kshemkalyani:2008:DCP:1374804,matocha1998taxonomy}. However, none of the existing TD protocols for static networks can be implemented straight forwardly in dynamic networks due to frequent topology changes in dynamic networks. Although, some TD protocols~\cite{DBLP:journals/tpds/TsengT01,DBLP:journals/tpds/DeMaraTE07,DBLP:conf/isads/KurianRS09,DBLP:conf/icdcs/JohnsonM09,6014939} exist for sensor networks and mobile ad hoc networks, they can also not be implemented in \textit{CRN}s due to unique challenges of cognitive radio networks, as mentioned in Section~\ref{subsec:Termination detection in the cognitive radio networks}.

A novel algorithm for TD using credit distribution and aggregation was proposed by Mattern~\cite{DBLP:journals/ipl/Mattern89} and Huang~\cite{DBLP:conf/icdcs/Huang89,DBLP:journals/jise/HuangK91}. A similar TD protocol for faulty distributed systems was proposed by Tseng~\cite{DBLP:journals/jpdc/Tseng95}. However, these protocols failed to work in dynamic networks. A TD protocol for mobile cellular networks~\cite{DBLP:journals/tpds/TsengT01} based on credit distribution and aggregation is proposed that assumes the existence of the mobile switching center (MSS), which provides a centralize support to the mobile nodes.

Johnson and Mittal~\cite{DBLP:conf/icdcs/JohnsonM09} have tried to reduce the waiting time for termination declaration in dynamic networks. However, they consider the existence of an initiator node until termination declaration. The protocols proposed for dynamic networks~\cite{DBLP:journals/tpds/TsengT01,DBLP:conf/isads/KurianRS09,DBLP:conf/icdcs/JohnsonM09,6014939} have three major limitations: (\textit{i}) they assume the existence of an initiator node until termination declaration; however, the mandatory existence of the initiator node increases the waiting time for the node that has completed its computation earlier than other nodes in the network. Also, the existence of an initiator node until termination declaration is not easy to guarantee in \textit{CRN}s, (\textit{ii}) they work on a single pre-decided channel, whereas, in \textit{CRN}, computations and nodes work on multi-channels, and (\textit{iii}) they consider only node mobility; they do not consider the presence of some special users (like primary users) that also prevent the nodes to work.

In \textit{CRN}, Mittal et al.~\cite{DBLP:conf/ispdc/ZengMVC10} presents a neighbor discovery protocol with TD; however, they consider only termination of the particular neighbor discovery scheme. The lightweight termination detection of Mittal et al.~\cite{DBLP:conf/ispdc/ZengMVC10} is not related to our termination detection scheme. Note that in~\cite{DBLP:conf/ispdc/ZengMVC10}, the term \enquote{lightweight} has been used to highlight the fact that the number of control messages used in their protocol is minimal.

\section{The System Settings}
\label{section:the_system_model}
This section outlines the preliminary assumptions about the environment, various types of messages (Table~\ref{tab:messages}), and data structures (Table~\ref{table:datastructure}). All the notations used in our protocol are given in Table~\ref{table:notations}.

\begin{description}
\item[Cognitive radio nodes.] We consider a cognitive radio network of $N$ cognitive radio nodes ($\mathit{CR}_1, \mathit{CR}_2, \ldots, \mathit{CR}_N$), where each node has a unique identity. However, a group of $n$ CRs executes a single computation, where $n \leq N$, in finite time. The nodes are heterogeneous in terms of their computing capabilities, and they are allowed to move during protocol execution.

    Each CR is aware of \textit{global channel set}, \textit{local channel set}, to be defined soon, and also the total number of nodes, $N$, in the network. Each node has a \textit{scan transceiver} (a transceiver is a transmitter-receiver pair) that is responsible for scanning multiple heterogeneous channels. Such a scanning is beneficial for fast channel switching. However, a transceiver cannot transmit and receive simultaneously.

\item[Communication channels.] We divide communication channels into two sets: (\textit{i}) \textit{global channel set} ($\mathit{GCS}$): a set of all the, $g$, channels in the network, where $g = |GCS|$; (\textit{ii}) \textit{local channel set} ($\mathit{LCS}$): a set of, $l$, available channels\footnote{Recall that a channel that is not currently occupied by a primary user is known as an available channel.} at a node, $\mathit{CR}_i$, where $l_i=|LCS_i|$ and $l_i\leq g$. However, the appearance of PU(s) on all $g$ channels results in the value of the local channel set to be zero, at each node. On the appearance of a PU, a CR is assumed to tune to another available channel, from its $\mathit{LCS}$, without interrupting the ongoing computation~\cite{Akyildiz:2006:NGS:1162469.1162470}, similar to the handoff in mobile cellular networks.

    A node, $\mathit{CR}_i$, that does not possess any available channel in its $LCS_i$ (\textit{i}.\textit{e}., $l_i = 0$) due to the appearance of PU(s), is called an \textit{affected node}. An affected node is unable to send and receive messages. On the other hand, a node, $\mathit{CR}_j$, that has at least one available channel in its $LCS_j$ (\textit{i}.\textit{e}., $l_j \geq 1$) is called a \textit{non-affected node}. The communication channels are non-FIFO (first-in-first-out) and unreliable. However, the sent messages must be received at the receiver nodes without omissions, duplications, and in the same order as they were sent~\cite{DBLP:conf/sss/DolevHSS12}, if the receiver is not an affected or a failed node (see failure model for details).

\item[Network structure.] We consider an asynchronous multi-hop cognitive radio network of $N$ independent nodes. We represent the network by a \textit{communication graph}, $\mathcal{CG} = [\mathcal{V}, \mathcal{E}, LCS]$. In the communication graph, $\mathcal{CG}$, $\mathcal{V}$ represents a set of vertices (or processors in the network), $\mathcal{E}$ represents a set of edges where an edge between a pair of neighboring nodes shows a bidirectional, direct, and non-FIFO wireless communication link, and $\mathit{LCS}$ represents the local channel set of each CR.

 Further, we define an \textit{interaction graph} of size $n \leq N$ as: $\mathcal{IG} = [v, e]$. In the interaction graph, $\mathcal{IG}$, $v \subseteq \mathcal{V}$ represents a set of CRs that are currently executing an identical computation and $e \subseteq \mathcal{E}$ represents a set of edges where each edge connects any two neighboring nodes, $\mathit{CR}_i, \mathit{CR}_j \in v$, if they are executing an identical computation. Note that we assume different interaction graphs for different computations.

\item[Failure model.] We assume that a cognitive radio node may fail in three different ways, as follows:
\begin{enumerate}[noitemsep,nolistsep]
  \item Due to the appearance of a PU and the node has only a single channel in its $\mathit{LCS}$, then the node is unable to send and receive messages, and such a node is called an affected node.
  \item Due to the swift movement of the node that may result in frequent topology change and transient non-interaction of the highly mobile node with other nodes in the network. We call such nodes the \textit{failed nodes}.
  \item Crash, \textit{i}.\textit{e}., when a node does not possess enough resources, like battery and computing power, it results in permanent failure of the node, and such a node is called a \textit{crashed node}. When a crashed node recovers by users' intervention, it does not possess the knowledge of updated data structures.
\end{enumerate}
In this protocol, we focus on the impact of PUs on the nodes, and after that the recovery of such nodes when PUs disappear. We do not consider any specific approach for recovery of failed nodes. The approach that works in MANET to handle failed nodes is also applicable in \emph{CRN}. In other words, we consider the \textit{failure-recovery model}~\cite{DBLP:journals/dc/AguileraCT00}. Whenever a node recovers, its state may be active or passive. It is possible that the failures occur frequently and, thereafter, the nodes recover soon. Such frequent failures and recoveries are not useful for any practical application; hence, we do not focus on these issues in our protocol. In addition, we assume that the affected and crashed nodes are detected by at least one of the nodes, whose state is active. We also assume that the nodes do not exhibit Byzantine behavior.

\item[Storage media.] The termination cannot be detected as the decision variable itself can be corrupted by transient failures leading to a false detection; hence, we store all the data structures in the non-volatile storage (\textit{i}.\textit{e}., stable storage). However, a consistent copy of the data is always available in the volatile memory. In the beginning, all the data structures are initialized.

\end{description}

\begin{table}[t]
\begin{center}
\begin{spacing}{0.85}
            \begin{tabular}{|l|l||l|l|}
            \hline
               {\scriptsize \textit{CRN}} & {\scriptsize Cognitive radio network} & {\scriptsize $N$} & {\scriptsize The total number of cognitive radio nodes} \\ \hline

               {\scriptsize $\mathcal{CG}$} & {\scriptsize Communication graph} & {\scriptsize $\mathcal{IG}$} & {\scriptsize Interaction graph} \\ \hline

               {\scriptsize $\mathcal{V}$} & {\scriptsize A set of cognitive radio nodes} & {\scriptsize $\mathcal{E}$} & {\scriptsize A set of edges between neighboring nodes} \\ \hline

               {\scriptsize $v$} & {\scriptsize A set of cognitive radio nodes in an interaction graph} &  {\scriptsize $e$} & {\scriptsize A set of edges in an interaction graph} \\ \hline

               {\scriptsize $C_E$} & {\scriptsize Chief executive node (or initiator of the protocol)} & {\scriptsize $n$} & {\scriptsize Number of nodes involved in an identical computation} \\ \hline

               {\scriptsize $\mathit{LCS}$} & {\scriptsize Local channel set} & {\scriptsize $\mathit{GCS}$} & {\scriptsize Global channel set} \\ \hline

               {\scriptsize $l$} & {\scriptsize Number of available channels at a node} & {\scriptsize $g$} & {\scriptsize Number of the channels in $\mathit{GCS}$} \\ \hline

            \end{tabular}
\B
\caption{Notations.}
\BBB\BBB
\label{table:notations}
\end{spacing}
\end{center}
\end{table}

\begin{table}
\begin{center}
\begin{spacing}{0.85}
   \begin{tabular}{ | p{2.8cm} | l | p{10cm} |   }
    \multicolumn{3}{c}{{\scriptsize Control messages}} \\ \hline

    {\scriptsize \textit{COMputation message}}               & {\scriptsize $\mathit{COM}(C) $}               & {\scriptsize send by $\mathit{CR}_i$ to its active/passive neighboring nodes to distribute the computation.} \\ \hline

    {\scriptsize \textit{I am Passive with Credit message}}  & {\scriptsize $\mathit{ImPC}(C, b)$}           & {\scriptsize send by $\mathit{CR}_i$ to all the active nodes that had sent credits to $\mathit{CR}_i$ previously including the parent node of $\mathit{CR}_i$. An $\mathit{ImPC}(C,b)$ message contains credit information, $C$, and the total number of active child nodes, $b$, of the sender $\mathit{CR}_i$. The value of $b$ is set to 0 if the $ImPC$ is sent to nodes other than the parent nodes.} \\ \hline

    {\scriptsize \textit{I am Passive message}}           & {\scriptsize $\mathit{ImP}(p)$}                  & {\scriptsize send by $\mathit{CR}_i$ to all its child nodes piggybacked with a new parent's, $\mathit{CR}_p$, information.} \\ \hline

    {\scriptsize \textit{AcKnowledgement message}}           & {\scriptsize $\mathit{AcK}$   }               & {\scriptsize send by $\mathit{CR}_j$ to $\mathit{CR}_i$ in order to acknowledge credit receipt, if $\mathit{CR}_i$ has surrendered its credit to $\mathit{CR}_j$.} \\ \hline

    {\scriptsize \textit{Acknowledgement of $AcK$ message}}    & {\scriptsize $AAcK$ }                & {\scriptsize send by $\mathit{CR}_i$ to $\mathit{CR}_j$ if $\mathit{CR}_i$ has received an $AcK$ message from $\mathit{CR}_j$. The highest priority messages, \textit{i}.\textit{e}., $AcK$ and $AAcK$, provide a three way handshake when $\mathit{CR}_i$ surrenders its credit to $\mathit{CR}_j$. The reception of $AcK$ and $AAcK$ messages are assumed to be atomic (and the delivery time of $AcK$ and $AAcK$ messages is very small, unlike other control messages).} \\ \hline

    {\scriptsize \textit{Termination Message}}               & {\scriptsize $\mathit{TM}$}                   & {\scriptsize send by the chief executive node, $C_E$, to all the nodes of the interaction graph to declare termination of the computation.} \\ \hline

    \multicolumn{3}{c}{{\scriptsize Non-control messages}} \\ \hline

    {\scriptsize \textit{Primary user affected Nodes message}}    & {\scriptsize $\mathit{PaN}$ }                 & {\scriptsize send by $\mathit{CR}_i$ that is neighboring node of $\mathit{CR}_j$ to $C_E$. This message holds the identity of the affected node, $\mathit{CR}_j$, and the credit that had sent to $\mathit{CR}_j$ by $\mathit{CR}_i$ or from $\mathit{CR}_i$ to $\mathit{CR}_j$.} \\ \hline

    {\scriptsize \textit{Nodes released by Primary user message}} & {\scriptsize $\mathit{NaP}$ }                 & {\scriptsize send by $\mathit{CR}_i$ to $C_E$ and all its neighbors whose states are active. A $\mathit{NaP}$ message holds the identity of $\mathit{CR}_i$, that was an affected node earlier; however, now $\mathit{CR}_i$ is a non-affected node.}  \\ \hline

    \multicolumn{3}{p{11cm}}
    {{\scriptsize Remarks: (\textit{i}) We use a notation $\mathit{SEND}_i(m, j)$ to show the message transmission of $m$ from $\mathit{CR}_i$ to $\mathit{CR}_j$.
    (\textit{ii}) The message transmission is shown in Figures~\ref{fig:abstract_model_without_PU} and~\ref{fig:abstract_model_with_PU}.}}

    \\ \hline
    \end{tabular}
    \B
    \caption{Message types used in the \textit{T-CRAN} protocol.}
    \BBB\BBB
    \label{tab:messages}
\end{spacing}
\end{center}
\end{table}

Furthermore, we assume that the cognitive radio nodes have sufficient battery and computing power, and an appropriate routing protocol is in place for message delivery. For ease of presentation and understanding, we consider a single instance of a single computation (\textit{i}.\textit{e}., a single interaction graph, $\mathcal{IG}$) in the network; however, the proposed protocol is able to handle multiple instances of multiple computations. We do not specify any neighbor discovery protocol; however, we assume that each CR knows its neighboring nodes using some existing neighbor discovery protocols, \textit{e}.\textit{g}.,~\cite{DBLP:journals/jpdc/MittalKCVZ09}.

\medskip\noindent \textbf{Types of messages}.
In our protocol, we use various messages (message details are given in Table~\ref{tab:messages}, and a simplified illustration of messages transmission is shown in Figures~\ref{fig:abstract_model_without_PU} and~\ref{fig:abstract_model_with_PU}) that are classified into \textit{control messages} and \textit{non-control messages}. The control messages have the highest transmission priority, and they are forwarded by (intermediate) passive nodes too. It is worth noting that only control messages require communication cost, and non-control messages can be piggybacked on the control or heartbeat messages.

All the control and non-control messages include a tuple $\langle session, initiator\_id\rangle$, that (\textit{i}) avoids the need of a logical clock, which is hard to implement in \textit{CRN}, (\textit{ii}) distinguishes any two messages, (\textit{iii}) distinguishes a message from \textit{stale messages} (a message that is received after the termination declaration, and so belongs to the terminated computation, is known as a stale message, throughout the paper).

\medskip\noindent \textbf{Types of data structures}. In our protocol, the data structures are divided into two categories: (\textit{i}) at all the cognitive radio nodes, and (\textit{ii}) at the chief executive node, $C_E$, (a node that is responsible for the announcement of termination). Details of these data structures are given in Table~\ref{table:datastructure}.

\section{Background}
\label{section:background}
A large number of termination detection (TD) protocols have been introduced for fault-free and faulty distributed systems. They are based on different scheme, \textit{e}.\textit{g}., snapshot, credit distribution and aggregation, logical tree, and ring structures~\cite{matocha1998taxonomy}. The snapshot based TD protocols require complex data structures to be maintained at each participating node because the amount of information exchanged is usually very high. Consequently, they have a large waiting time for the announcement of termination that is unsuitable for ad hoc networks. On the other hand, the maintenance of logical structures (rings and trees) is a computation intensive task, and it requires frequent exchange of coordination messages to handle dynamic topology in ad hoc environment. Such a high overhead is deterrent in the use of logical structures. Therefore, we consider the credit distribution and aggregation approach to design a TD protocol for \textit{CRN}.

For the sake of completeness and understanding of credit distribution and aggregation based TD protocols, we present the first credit distribution and aggregation based TD protocol, given by Mattern~\cite{DBLP:journals/ipl/Mattern89}. This protocol assumes the existence of an \textit{oracle} that is responsible for initiation of the computation and termination detection.

Initially, the network has all the nodes in passive state, and the credit at each node is zero. In the beginning of a computation, the oracle, which is supposed to have credit value 1, distributes the credit value among the nodes using \textit{activation messages}. Thus, each activation message holds a credit value, $0<C<1$. Now, the oracle waits to receive credits back. Once, the cumulated credit has value one, the oracle announces termination. This protocol uses four rules, as follows:
\begin{table}
\begin{center}
\begin{spacing}{0.85}
            \begin{tabular}{ | l | p{9.2cm} | l|}
            \hline
               {\scriptsize Data structure} & {\scriptsize Description} & {\scriptsize Initial value} \\ \hline

               \multicolumn{3}{c}{{\scriptsize At all the cognitive radio nodes}} \\ \hline

               {\scriptsize $parent_i$} & {\scriptsize The parent node of $\mathit{CR}_i$. It is the first node that sent credit to $\mathit{CR}_i$ since $\mathit{CR}_i$ became active.} & {\scriptsize 0} \\ \hline

               {\scriptsize $\mathit{hold}_i$} & {\scriptsize The credit received from the parent node of $\mathit{CR}_i$.} & {\scriptsize 0} \\ \hline

               {\scriptsize $in_i[]$} & {\scriptsize $in_i[j]$ represents the received credit at $\mathit{CR}_i$ from $\mathit{CR}_j$ such that $j \neq parent_i$.} & {\scriptsize $\forall j, in_i[j] = \emptyset$}\\ \hline

               {\scriptsize $out_i[]$} & {\scriptsize $out_i[j]$ represents the credit sent from $\mathit{CR}_i$ to $\mathit{CR}_j$.} & {\scriptsize $\forall j, out_i[j] = \emptyset$}\\ \hline

               {\scriptsize $\mathit{session}_i$} & {\scriptsize The current session of the computation at $\mathit{CR}_i$.} & {\scriptsize 0} \\ \hline

               {\scriptsize $\mathit{initiator}\_id$} & {\scriptsize The initiator of the current session of the computation at $\mathit{CR}_i$.} & {\scriptsize 0} \\ \hline

               \multicolumn{3}{c}{{\scriptsize At the chief executive node, $C_E$}} \\ \hline

               {\scriptsize $PU_{\mathit{affected}}[]$} & {\scriptsize $PU_{\mathit{affected}}[j]$ represents the identity of an affected node, $\mathit{CR}_j$.} & {\scriptsize $\forall j, PU_{\mathit{affected}}[j] = \emptyset$}\\ \hline

               {\scriptsize $\mathit{C\_PU_{\mathit{affected}}[]}$} & {\scriptsize $C\_PU_{\mathit{affected}}[j]$ represents the credit at an affected node, $\mathit{CR}_j$.} & {\scriptsize $\forall j, C\_PU_{\mathit{affected}}[j] = \emptyset$}\\ \hline

            \end{tabular}
\B
\caption{Data structures used in the \textit{T-CRAN} protocol.}
\BBB\BBB
\label{table:datastructure}
\end{spacing}
\end{center}
\end{table}

\begin{description}[noitemsep]
  \item[R1] When a passive node receives an activation message with credit $0<C<1$, the node becomes active, holds the credit $C$, and executes the assigned computation.
  \item[R2] When an active node receives an activation message with credit $0<C<1$, the credit value, $C$, is transferred to the oracle.
  \item[R3] When an active node, having credit $C$, sends an activation message, the node sends only $\frac{C}{2}$ credit with the message.
  \item[R4] When a node becomes passive, it surrenders its credit to the oracle.

\end{description}
\begin{wrapfigure}{r}{8cm}
\BBB\BBB
\begin{center}
\begin{minipage}{0.45\textwidth}
          \includegraphics{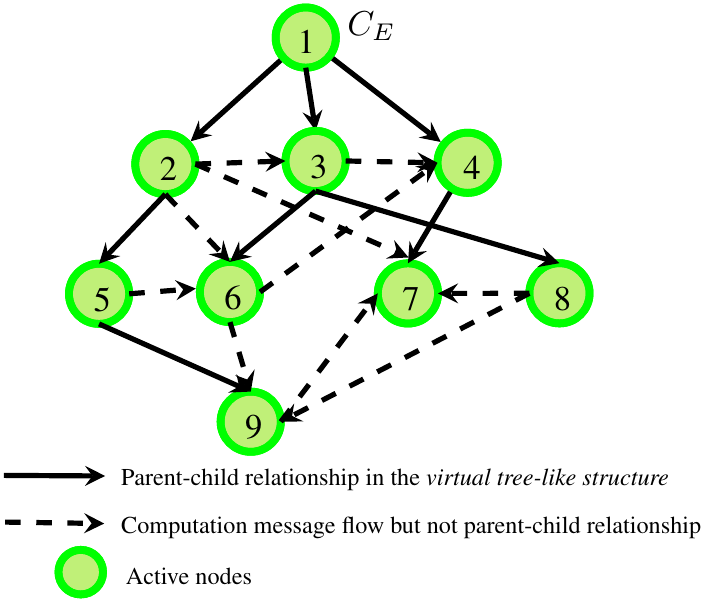}
\end{minipage}
\end{center}
\BBB\B
\caption{The \textit{virtual tree-like structure}.}
\BBB\BBB
\label{fig:The virtual tree-like structure.}
\end{wrapfigure}

In addition, this protocol always satisfies the three requirements: (\textit{i}) at any time, the sum of credits held by nodes, activation messages, and the oracle is 1, (\textit{ii}) when a node is active, it holds a credit $C> 0$, and (\textit{iii}) an activation message, which is in-transit, holds a credit $C> 0$.

The limitations of Mattern's protocol~\cite{DBLP:journals/ipl/Mattern89} is the existence of a fixed oracle to announce termination that increases the waiting time for termination announcement. Also, this protocol is assumed to work in static networks, where the nodes communicate using fixed communication links. However, unlike Mattern's protocol~\cite{DBLP:journals/ipl/Mattern89}, our protocol is designed for dynamic \textit{CRN}s that have multiple heterogeneous channels. Also, we do not assume the existence of a fixed oracle. Further details of the proposed protocol are presented in the next section.

\section{The \textit{T-CRAN} Protocol}
\label{section:The T-CRAN Protocol}


We present our credit distribution and aggregation based termination detection protocol, called \textit{T-CRAN} (see Figures~\ref{fig:abstract_model_without_PU} and~\ref{fig:abstract_model_with_PU}). The initiation of \emph{T-CRAN} protocol is marked by the distribution of a fixed credit value, $C$, and when a node receives back the same credit value, $C$, it announces termination.

\subsection{High level description of the \textit{T-CRAN} protocol}
\label{subsec:High level description of our protocol}
A node initiates a computation and the \textit{T-CRAN} protocol\footnote{Recall that the \textit{T-CRAN} protocol is modeled as another layer on top of the computation; hence, it executes concurrently with the computation.} with a \textit{fixed} credit value, $C$, and such a node is called the \textit{chief executive node}, $C_E$. $C_E$ may distribute the computation among its neighboring nodes, called the \textit{child nodes} (of $C_E$), with non-zero credit values, and $C_E$ becomes a \textit{parent node} of its \emph{child nodes}. The \emph{child nodes} can further distribute the computation like their \emph{parent node}. In this manner, the credit distribution phase creates an illusion of a logical tree among the CRs that are executing an identical computation. We call it the \textit{virtual tree-like structure}, henceforth (see Figure~\ref{fig:The virtual tree-like structure.}). Note that the sum of credits in the network (including the nodes and in-transit messages) must be equal to $C$.

When a node finishes its computation, the node's state becomes passive, and the node surrenders its credit. In the \textit{virtual tree-like structure}, a node surrenders its credit to either (\textit{i}) its parent node if the parent node is active, (\textit{ii}) any node whose state is active and that had sent credit to the node previously, or (\textit{iii}) any node that is executing the same computation, whose state is active.\footnote{Such a credit surrender process decreases the waiting time for any node, especially for parent nodes and $C_E$, if they have finished their computation earlier than their child nodes. Following that it is clear that the parent nodes are also allowed to surrender their credit to any of their child nodes if they are active or to any active neighboring node that is executing the identical computation.} As PUs appear, the neighboring nodes\footnote{A preference is given to neighboring nodes whose states are active.} of the affected nodes inform $C_E$ about the affected nodes. Once the affected nodes become non-affected nodes, they inform about their recovery to $C_E$ and their neighboring nodes, whose states are still active. However, $C_E$ waits for a reasonable amount of time\footnote{The time may be based on the size of the network, message transmission time, and criticality of the computation.} for the transition of affected nodes to non-affected nodes before the announcement of termination. Once $C_E$ receives back the credit $C$ (the same amount of credit that was distributed at the initiation of the computation) or a timeout occurs, it announces termination.

\noindent \textbf{Comparison with conventional credit distribution and aggregation protocols.} The difference between the conventional credit distribution and aggregation protocols~\cite{DBLP:journals/ipl/Mattern89,DBLP:journals/jise/HuangK91,DBLP:journals/tpds/TsengT01,DBLP:conf/icdcs/JohnsonM09} and our protocol lies in the credit surrender process when a node completes its computation.

The conventional credit distribution and aggregation protocols use a logical tree structure, where a \textit{fixed} root node announces termination; thus, it is mandatory that the root node stays in active state till the end of the computation.

Our credit distribution and aggregation based protocol, \textit{T-CRAN}, uses a logical structure, called the \textit{virtual tree-like structure}, where a \textit{non-fixed} $C_E$ may become passive on completion of its computation, and may send its credit, arbitrarily, to one of its child node that becomes a new $C_E$. The new $C_E$ is responsible for termination declaration, and thus, the existence of an identical $C_E$ till the end of the computation is not desired in the \textit{virtual tree-like structure}.

\begin{figure*}
\begin{framed}
\begin{center}
    \begin{minipage}[b]{.5\textwidth}
        \centering
          \includegraphics{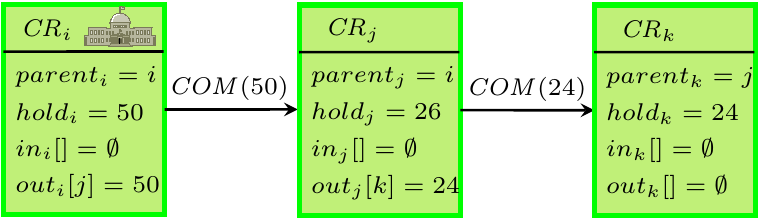}
       \subcaption{}
       \label{fig:abstract_model_a}
    \end{minipage}
    \begin{minipage}[b]{.49\textwidth}
        \centering
         \includegraphics{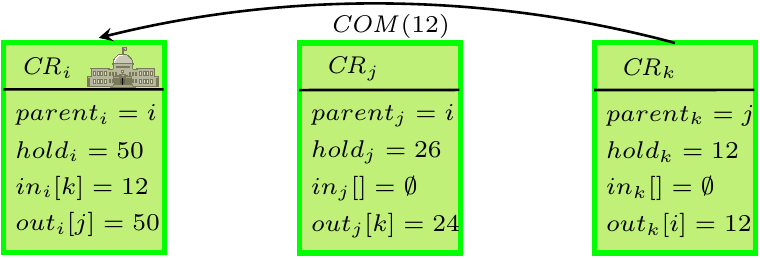}
        \subcaption{}
        \label{fig:abstract_model_b}
    \end{minipage}
     \begin{minipage}[b]{.5\textwidth}
        \centering
          \includegraphics{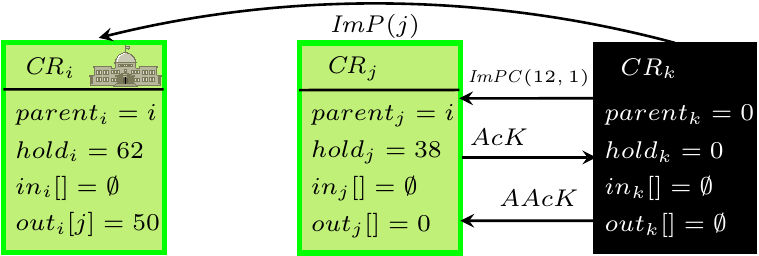}
       \subcaption{}
       \label{fig:abstract_model_c}
    \end{minipage}
    \begin{minipage}[b]{.49\textwidth}
        \centering
         \includegraphics{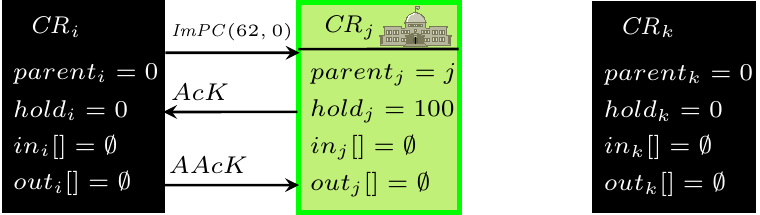}
        \subcaption{}
        \label{fig:abstract_model_d}
    \end{minipage}
     \begin{minipage}[b]{.5\textwidth}
        \centering
          \includegraphics{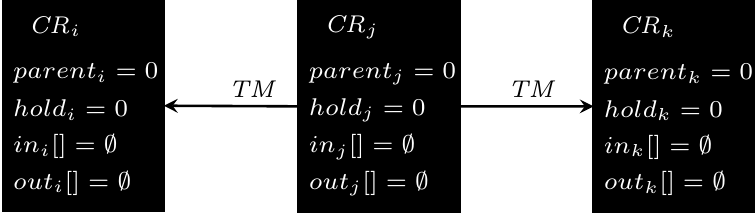}
       \subcaption{}
       \label{fig:abstract_model_e}
   \end{minipage}
     \begin{minipage}[b]{.49\textwidth}
        \centering
            \includegraphics{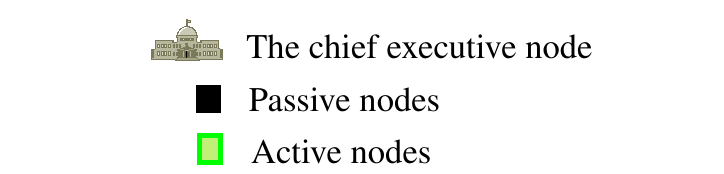}
       \subcaption{}
       \label{fig:abstract_model_f}
   \end{minipage}
\end{center}

Figures~\ref{fig:abstract_model_a} and~\ref{fig:abstract_model_b} show the \textit{credit distribution phase} (\textsc{Step}s 1-3, \textsc{Step}s 1-6 are given in Section~\ref{subsec:Details of the T-CRAN Protocol}). In Figure~\ref{fig:abstract_model_a} (\textsc{Step}s 1-2), $\mathit{CR}_i$ initiates a computation and the \textit{T-CRAN} protocol; hence, $\mathit{CR}_i$ becomes the chief executive node. $\mathit{CR}_j$ receives a \textit{COMputation message} from $\mathit{CR}_i$ with credit 50, holds a credit value ($hold_j=26$), and further distributes the computation to $\mathit{CR}_k$. In Figure~\ref{fig:abstract_model_b} (\textsc{Step} 3), $\mathit{CR}_k$ distributes the computation to $\mathit{CR}_i$, where $\mathit{CR}_i$ does not initialize $\mathit{CR}_k$ as its parent node.
~\\
Figures~\ref{fig:abstract_model_c} and~\ref{fig:abstract_model_d} show the \textit{credit aggregation phase} (\textsc{Step}s 4-5). In Figure~\ref{fig:abstract_model_c} (\textsc{Step}s 4-5), $\mathit{CR}_k$ becomes passive and surrenders its credit to its parent $\mathit{CR}_j$ (using an \textit{I am Passive with Credit message}, $\mathit{ImPC}(C,z)$). $\mathit{CR}_k$ also sends an \textit{I am Passive message}, $\mathit{ImP}(p)$, to $\mathit{CR}_i$, where the child node of $\mathit{CR}_k$ (\textit{i}.\textit{e}., $\mathit{CR}_i$) has not terminated its computation. Also, $\mathit{CR}_j$ and $\mathit{CR}_k$ do a three way handshake to ensure a lossless delivery of the credit at $\mathit{CR}_j$ (using an \textit{AcKnowledgement message}, $AcK$, and an \textit{Acknowledgement of AcK message}, $AAcK$). In Figure~\ref{fig:abstract_model_d}, the chief executive node, $\mathit{CR}_i$, surrenders its credit to $\mathit{CR}_j$, and $\mathit{CR}_i$, $\mathit{CR}_j$ do a three way handshake. Now, $\mathit{CR}_j$ becomes the new chief executive node.
~\\
Figure~\ref{fig:abstract_model_e} shows termination announcement (\textsc{Step} 6), where the chief executive node, $\mathit{CR}_j$, announces termination using \textit{Termination Message}s, $\mathit{TM}$, when $\mathit{CR}_j$ becomes passive.
\end{framed}
\BBB
\caption{The termination detection protocol, \textit{T-CRAN}, in the absence of primary users.}
\BBB
\label{fig:abstract_model_without_PU}
\end{figure*}

\subsection{Details of the \textit{T-CRAN} Protocol}
\label{subsec:Details of the T-CRAN Protocol}
Now, we first provide details of credit distribution and aggregation phases in the absence of PUs. Later in Section~\ref{subsec:Detection of Primary Users}, we will consider the presence of PUs too.

\medskip

\noindent \textbf{Credit distribution.} In our protocol, initially, all the nodes are in passive state, and a computation starts by a single message from outside world. The credit distribution phase creates a \textit{virtual tree-like structure} (see Figure~\ref{fig:The virtual tree-like structure.}) and consists of three steps (see Figures~\ref{fig:abstract_model_a} and~\ref{fig:abstract_model_b}), as follows:
\begin{description}[noitemsep,nolistsep]
  \item[\textsc{Step} 1: Initiation and distribution of a computation and the \textit{T-CRAN} protocol.] The computation and the \textit{T-CRAN} protocol is initiated by a CR node, with a fixed credit value $C$ (that is stored in variable $hold_{C_{E}}$), called the \textit{chief executive node}, $C_E$. The node $C_E$ may distribute the computation among its $q$ ($q\leq N$) neighboring CRs with non-zero credit values, say $C_1, C_2, \ldots, C_q$, using $q$ different \textit{COMputation message}s ($\mathit{COM}(C_1), COM(C_2), \ldots, COM(C_q)$). Note that once the credit distribution is over, the total credit in the network must be $C$, \textit{i}.\textit{e}., $hold_{C_{E}}+C_1+C_2+\ldots+C_q =C$. Also, we assume that the division of any credit value do not result in a floating point problem, which may result in fractional loss of credits.

  \item[\textsc{Step} 2: Reception of a \textit{COMputation message} at a passive node.] The reception of a $\mathit{COM}(C_j)$ at a passive node, say $\mathit{CR}_j$, causes $\mathit{CR}_j$ to become active and initiate the computation. In addition, $\mathit{CR}_j$ holds credit $C_j$ (in $hold_j$ variable). $\mathit{CR}_j$ may also distribute the computation among neighboring node(s) with non-zero credit values (following the procedure similar to $C_E$). Secondly, on reception of the first \textit{COMputation message} from any node, say $\mathit{CR}_x$, at a node, say $\mathit{CR}_y$, the node $\mathit{CR}_y$ designates the node $\mathit{CR}_x$ as its parent node and becomes a child of $\mathit{CR}_x$.

  \item[\textsc{Step} 3: Reception of \textit{COMputation message}s at an active node.] An active node, $\mathit{CR}_j$, may receive further \textit{COMputation message}s from the nodes other than its parent, say from $\mathit{CR}_k$. In such a situation, $\mathit{CR}_k$ does not become the parent node of $\mathit{CR}_j$. Also, the newly received credit value does not increase credit that $\mathit{CR}_j$ holds (\textit{i}.\textit{e}., $hold_j$). However, $\mathit{CR}_j$ performs the corresponding computation and keeps the received credit in an array $in_j[i]$.
\end{description}
\begin{figure*}
\begin{framed}
\begin{center}
    \begin{minipage}{.5\textwidth}
        \centering
          \includegraphics[width=80mm, height=25mm]{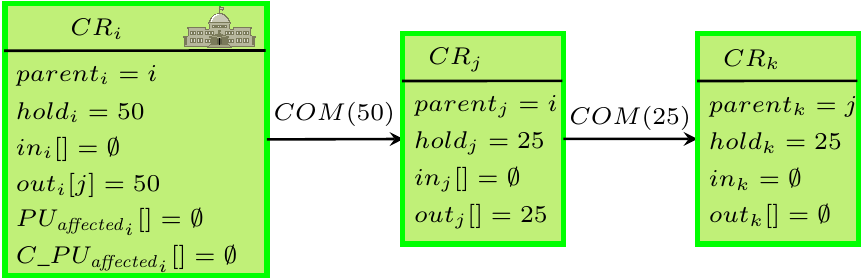}
       \subcaption{}
       \label{fig:abstract_model_PU_a}
    \end{minipage}
    \begin{minipage}{.49\textwidth}
        \centering
         \includegraphics[width=80mm, height=25mm]{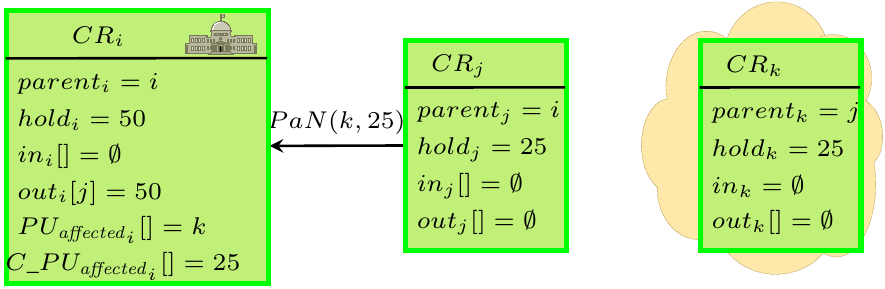}
        \subcaption{}
        \label{fig:abstract_model_PU_b}
    \end{minipage}
     \begin{minipage}[b]{.5\textwidth}
        \centering
          \includegraphics[width=80mm, height=25mm]{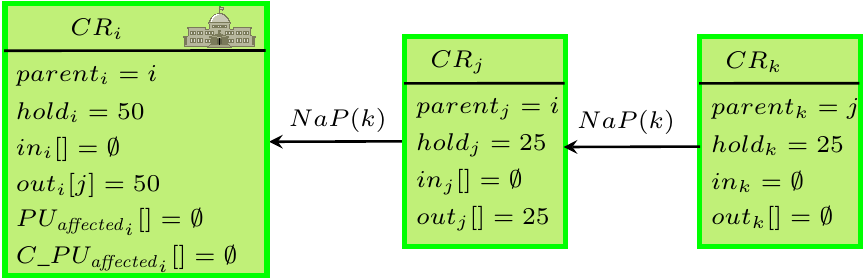}
       \subcaption{}
       \label{fig:abstract_model_PU_c}
    \end{minipage}
    \begin{minipage}[b]{.49\textwidth}
        \centering
         \includegraphics{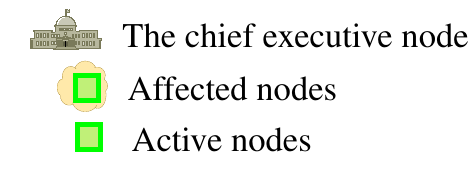}
        \subcaption{}
        \label{fig:abstract_model_PU_d}
    \end{minipage}
\end{center}
Figure~\ref{fig:abstract_model_PU_a} shows the \emph{credit distribution phase} (\textsc{Step}s 1-2), like in Figure~\ref{fig:abstract_model_a}, and $\mathit{CR}_i$ becomes the chief executive node, $C_E$.
~\\
Figure~\ref{fig:abstract_model_PU_b} shows the presence of a PU (\textsc{Step} 7, \textsc{Step}s 7-8 are given in Section~\ref{subsec:Detection of Primary Users}). $\mathit{CR}_k$ becomes an affected node that does not possess any available channel. However, $\mathit{CR}_i$ and $\mathit{CR}_j$ continues the computation because of some available channels. $\mathit{CR}_j$ sends a \textit{Primary user affected Nodes message} ($PaN(k,25)$) to $C_E$ to inform $\mathit{CR}_k$'s unavailability in the computation.
~\\
Figure~\ref{fig:abstract_model_PU_c} shows actions when a PU disappears (\textsc{Step} 8), and hence, $\mathit{CR}_k$ again becomes a non-affected node. $\mathit{CR}_k$ informs its neighbor, $\mathit{CR}_j$, and $C_E$ about recovery using \textit{Nodes released by Primary user messages}, $\mathit{NaP}$.
\end{framed}
\BBB
\caption{The termination detection protocol, \textit{T-CRAN}, in the presence of primary users.}
\BBB
\label{fig:abstract_model_with_PU}
\end{figure*}

\medskip

\noindent \textbf{Credit aggregation.} At the end of the credit distribution phase, the recipients of non-zero credits, become part of an interaction graph, $\mathcal{IG}$. When the nodes complete their computation, the credit aggregation phase is initiated. The credit aggregation phase (see Figures~\ref{fig:abstract_model_c} and~\ref{fig:abstract_model_d}) consists of two steps, as follows:
\begin{description}[noitemsep,nolistsep]
  \item[\textsc{Step} 4: Credit surrendering by active nodes.] Once $\mathit{CR}_j$ finishes its computation, it surrenders its credits to the corresponding nodes, whose states are active and had sent some credits to $\mathit{CR}_j$ previously. Unlike~\cite{DBLP:journals/ipl/Mattern89,DBLP:journals/jise/HuangK91,DBLP:journals/tpds/TsengT01,DBLP:conf/icdcs/JohnsonM09}, the \textit{T-CRAN} protocol elevates the credit surrender process at any node by allowing the node to surrender its credits, after the completion of its computation, to the corresponding sender nodes that are active (or vice versa). Moreover, $\mathit{CR}_j$ may surrender its credit to any node whose state is active and that is executing the identical computation, in case, its parent node and all the child nodes have become passive.

      The credit surrender process reduces the waiting time for any parent node (or child nodes) that wants to terminate its computation. In fact, $C_E$ can also surrender its credit to any of its child nodes, say $\mathit{CR}_x$, and $\mathit{CR}_x$ becomes the new $C_E$. In addition, $C_E$ also transfers its data structures, namely $PU_{\mathit{affected}}[]$ and $C\_PU_{\mathit{affected}}[]$ (see Table~\ref{table:datastructure}), to $\mathit{CR}_x$ at the time of credit surrender.

  \item[\textsc{Step} 5: Three-way handshake.] In the absence of failures, $\mathit{CR}_i$ acknowledges $\mathit{CR}_j$ (to inform $\mathit{CR}_j$ that $\mathit{CR}_i$ has received the credit back from $\mathit{CR}_j$) using an \textit{AcKnowledgement message} ($AcK$). $\mathit{CR}_j$ also acknowledges the reception of the $AcK$ to $\mathit{CR}_i$ using an \textit{Acknowledgement of AcK message} ($AAcK$). Such a mechanism provides a \textit{three-way handshake} and ensures safe delivery of credits. However, the non-reception of an $AcK$ at $\mathit{CR}_j$, after a timeout, causes $\mathit{CR}_j$ to surrender its credit to another node whose state is active. The three-way handshake can be avoided, if there is a guarantee of message delivery at the receiving node, which is neither an affected node nor a crashed node (using an algorithm suggested in~\cite{DBLP:conf/sss/DolevHSS12}).
%
\end{description}

\medskip

\noindent \textbf{Termination declaration.} In the beginning, any node may initiate the termination detection protocol and becomes $C_E$. However, once initiated, any node may take charge as $C_E$ (according to \textsc{Step} 4), and the new $C_E$ declares the final termination (see Figure~\ref{fig:abstract_model_e}). In fact, between initiation and termination of any computation, there could be multiple chief executive charge handovers in the network.
\begin{description}[noitemsep,nolistsep]
  \item[\textsc{Step} 6: Termination announcement.] When $C_E$ holds credit $C$ and it has completed its computation, $C_E$ informs all the other nodes in the interaction graph, $\mathcal{IG}$, about the termination of the computation using \textit{Termination Message}s ($\mathit{TM}$). However, in any case, only a single node (\textit{i}.\textit{e}., $C_E$) can announce termination (when it holds credit $C$). We also relax the termination detection criteria in Section~\ref{subsec:Termination Declaration}.
\end{description}

\subsection{Detection of primary user(s)}
\label{subsec:Detection of Primary Users}
The appearance of a PU perturbs the working of the CRs (as well as the termination detection protocol) and forces them to tune to another available channel in their $\mathit{LCS}$. Specifically, the appearance of a PU can be visualized similar to the network partitioning, as it partitions the network into two parts, as explained below:
\begin{itemize}[noitemsep]
  \item Primary user(s) affected \textit{CRN} ($CRN_P$): It is a part of the network that consists of affected nodes, which cannot send or receive any message.
  \item Non-primary user(s) affected \textit{CRN} ($CRN_N$): It consists of all the CRs that have completed their computation and surrendered their credit to the respective senders (of the credit).
\end{itemize}

\begin{figure}
\centering
          \includegraphics[height=30mm]{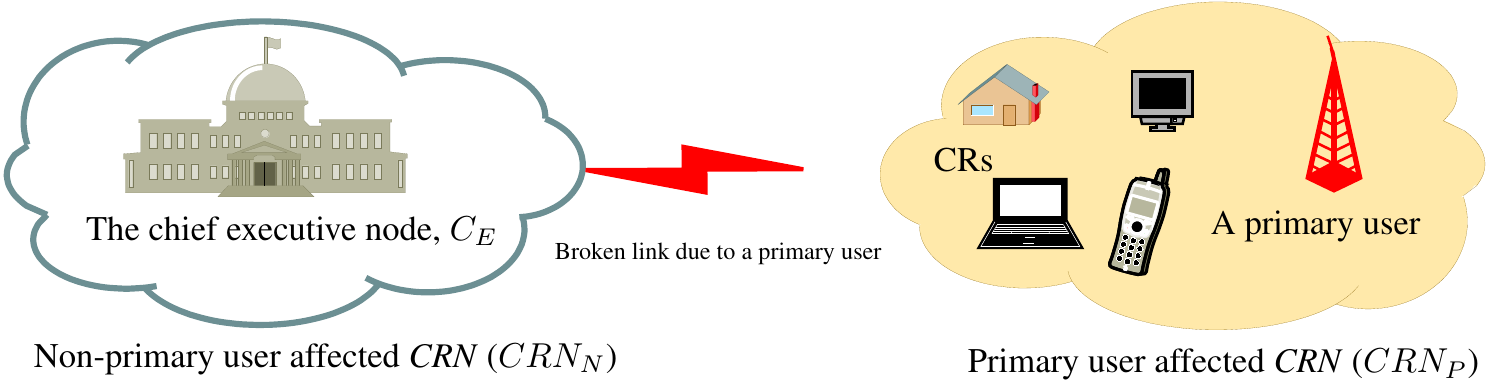}
\B
\caption{The abstract view of a primary user's appearance in cognitive radio networks.}
\BBB
\label{fig:The_abstract_view_of_a_primary_user appearance in the cognitive radio network.}
\end{figure}

In Figure~\ref{fig:The_abstract_view_of_a_primary_user appearance in the cognitive radio network.}, we show these two \textit{CRN}s, namely $CRN_P$ and $CRN_N$, where, the total credit $C$ is the sum of credits at $CRN_P$ and credit at $CRN_N$. Further, $C_E$ waits for credits of affected nodes. Intuitively, the appearance of a PU can be interpreted as follows:
\begin{itemize}[noitemsep]
  \item When a PU never ever leaves the channel and the nodes are unable to tune to another available channel, the state of the affected nodes can be interpreted as a crash (that is a permanent failure).
  \item When a PU persists in the network for a very long time, the computation at the affected nodes can be interpreted as excessively slowed down due to PUs.
\end{itemize}
However, both the above situations are indistinguishable for other nodes in the network. Thus, we develop an approach that is useful to declare termination even in the presence of primary users. Note that the available hardware approaches –-- match filter, energy filter, feature filter, inference temperature management~\cite{1399240} and spectrum sensing techniques~\cite{DBLP:journals/tmc/UrgaonkarN09,DBLP:journals/ejasp/ZengLHZ10,DBLP:journals/ejasp/WangLSH10} --– are capable enough to detect the presence of PUs (by any CR). Specifically, the appearance of a PU may turn a non-affected node to an affected node (see \textsc{Step} 7 and Figure~\ref{fig:abstract_model_PU_b}), and when the PU leaves the channel, an affected node becomes a non-affected node (see \textsc{Step} 8 and Figure~\ref{fig:abstract_model_PU_c}), as follows:
\begin{description}[noitemsep]
  \item[\textsc{Step} 7: Node failure due to the appearance of PUs.] An affected node, say $\mathit{CR}_j$, is detected by all the neighboring nodes, whose states are active (and these neighboring nodes may be the parent node or child nodes of $\mathit{CR}_j$). All the neighboring nodes of $\mathit{CR}_j$, whose states are active, inform $C_E$ about such a situation using a \textit{Primary user affected Nodes message} ($\mathit{PaN}$). Each $\mathit{PaN}$ holds the identity of $\mathit{CR}_j$, and the credit sent (received) to (from) $\mathit{CR}_j$. On reception of each $\mathit{PaN}$, $C_E$ enlists $\mathit{CR}_j$ and $\mathit{CR}_j$'s credit in the corresponding data structures (namely, $PU_{\mathit{affected}}[]$ and $C\_PU_{\mathit{affected}}[]$).

      Further, all the senders of $\mathit{PaN}$ messages remove credit information about $\mathit{CR}_j$ from their $in[]$ or $out[]$, whichever the case may be. Note that, like any CR, $C_E$ can also detect its affected neighboring nodes, and it does the same as on receiving a $\mathit{PaN}$ and removes them from $in_{C_{E}}[]$ or $out_{C_{E}}[]$. However, other $\mathit{PaN}$ messages, about the same affected nodes require the identical processing at $C_E$.

      Moreover, the reception of $\mathit{PaN}$ messages is sufficient to declare termination of the computation, after a reasonable amount of time, in case the following equation~\ref{eq:TD_when_affected_nodes} holds true:
      \begin{equation}\label{eq:TD_when_affected_nodes}
      out_{C_{E}}[]= \emptyset \wedge in_{C_{E}}[]= \emptyset \wedge (hold_{C_{E}} + C\_PU_{\mathit{affected}_{C_{E}}}[] = C)
      \end{equation}
      Such a termination detection is called \textit{weak termination} (see Section~\ref{subsec:Termination Declaration} for details, and the equation~\ref{eq:TD_when_affected_nodes} will be proved in Section~\ref{section:Safety property}).

  \item[\textsc{Step} 8: Recovery of the affected nodes.] Once an affected node, say $\mathit{CR}_j$, becomes a non-affected node, $\mathit{CR}_j$ informs $C_E$ and all its neighboring nodes, whose states are active, using \textit{Nodes released by Primary user messages} ($\mathit{NaP}$). On reception of a $\mathit{NaP}$ message at $C_E$, $C_E$ first checks whether the computation has terminated. If not, then $C_E$ informs $\mathit{CR}_j$ about the ongoing computation (with $\mathit{CR}_j$'s credit value). On the other hand, if the neighboring nodes of $\mathit{CR}_j$ have not completed the computation, they hold the credit back in their respective data structures, $in[]$ or $out[]$, whichever the case may be. In this manner, the total credit remains identical as it was at the time of computation initiation.
\end{description}

\subsection{Termination declaration}
\label{subsec:Termination Declaration}
Any kind of node failure (\textit{e}.\textit{g}., non-availability of a channel in $\mathit{LCS}$, mobility of the nodes, and crash, given in Failure Model, Section~\ref{section:the_system_model}) may lead to either temporary or permanent disconnection of the CRs from the network as well as discontinuity of the computation. However, the failures are quite common in \textit{CRN}. It is not surprising that such a disconnection may leave $C_E$ to starve to collect the necessary credits for termination declaration. Hence, in order to avoid the endless waiting at $C_E$, $C_E$ may declare either kind of termination, as defined below:
\begin{itemize}[noitemsep]
  \item \textit{Strong termination} infers passive state of all the CRs that were executing an identical computation and the absence of in-transit messages. It is also called \textit{correct and safe termination} announcement.

  \item \textit{Weak termination} refers to $CRN_P$ and $CRN_N$, where the state of all the non-affected CRs is passive, and there is no in-transit message. Also, the disappearance of PUs (or recovery from mobility) results in the strong termination. The main advantage of weak termination is the detection of affected nodes and to avoid endless waiting to announce strong termination. However, our approach announces weak termination if equation~\ref{eq:TD_when_affected_nodes} holds true.

      The significance of weak termination can be figured out by the following example: suppose, we start a leader election (LE) protocol~\cite{DBLP:conf/icdcsw/SharmaS11} in \textit{CRN} that consists of 100 nodes, initially. During the execution of the LE protocol, say 20 nodes became affected nodes. Thus, it is impractical to wait for strong termination, because a leader may also be elected out of 80 nodes. After recovery from PUs, the remaining 20 nodes may join the network. Such a scenario reduces the waiting time for the announcement of a leader, maintains computation continuity, and enhances resource utilization. However, the weak termination losses its significance, when it is unable to satisfy the safety requirements in the computation, \textit{e}.\textit{g}., mutual exclusion and consensus.
\end{itemize}

Two more criteria for termination in the network are also defined, as follows:
\begin{itemize}[noitemsep]
  \item \textit{Local termination} represents termination of the computation at a node. Further, the node has surrendered its credit to its parent, any neighbor, or any node that is executing the identical computation.

  \item \textit{Global termination} represents termination of the computation at all the nodes. In other words, the local termination at all the nodes may lead to the global termination, in case, no message is in-transit.
\end{itemize}
More specifically, two possible outcomes of termination are considerable as: global weak termination or global strong termination.

\begin{table}
\small
\begin{center}
    \begin{tabular}{| p{17cm} |}
    \hline
    \textbf{Notations:} $\mathit{SEND}_i(m, j)$: $\mathit{CR}_i$ sends a message $m$ to $\mathit{CR}_j$, $\mathit{STATE}(i)$: current state of $\mathit{CR}_i$, \textit{i}.\textit{e}., either active or passive, $t_e$: value of timeout. All the data structures have usual meanings (see Table~\ref{table:datastructure} for details of the data structures).\\ \\

    \textbf{$A_1$.} \textbf{Computation and protocol initiation.} $\mathit{CR}_i$ receives a message $M$ from outside world $\rightarrow$ \\
    \hspace{1em} $parent_i:= i, hold_i:= C, in_i[]:= \emptyset, out_i[]:= \emptyset, session_i:= x$ \\ \\

    \textbf{$A_2$.} \textbf{Distribution of credits.} $\mathit{CR}_i$ sends \textit{COMputation message}s to its $q$ ($q\leq N$) neighboring nodes$\rightarrow $ \\
    \hspace{1em} $parent_i:= i, hold_i:= C_i, in_i[]:= \emptyset, out_i[1,2,\ldots,q]:= \{C_1,C_2,\ldots,C_q\}, session_i:= x,$\\
    \hspace{1em} for all $q$ neighboring nodes $\mathit{SEND}_i(COM(C_p), p)$\\ \\

    \textbf{$A_3$.} \textbf{Reception of \textit{COMputation message}s.} $\mathit{CR}_j$ receives a \textit{COMputation message} ($\mathit{COM}(C_j)$) from $\mathit{CR}_i \rightarrow$ \\
    \hspace{1em} $\ [\!] parent_j= \emptyset \wedge \mathit{STATE}(j)=\mathit{PASSIVE} \rightarrow parent_j:= i, hold_j:= C_j, in_j[]:= \emptyset, out_j[]:= \emptyset, session_j:= x$\\
    \hspace{1em} $\ [\!] parent_j= \mathit{CR}_a \wedge \mathit{STATE}(j)=\mathit{ACTIVE} \wedge hold_j:= C_a \rightarrow parent_j:= \mathit{CR}_a, hold_j:= C_a, in_j[i]:= C_j$, \\
    \hspace{2em} $out_j[]:= \emptyset, session_j:= x$\\\\

    \textbf{$A_4$.} \textbf{Credit surrender.} $\mathit{CR}_j$  becomes idle $\rightarrow$ \\
    \hspace{1em} $\forall k: k \in in_j[] \wedge \mathit{STATE}(k)= \mathit{ACTIVE} \rightarrow \mathit{SEND}_i(ImPC(in_i[k], 0), k)$, \textit{Three-wayHandshake}$()$,\\

    \hspace{1em} $\ [\!] parent_j = j \rightarrow$ {\scriptsize (//A case to show when $\mathit{CR}_j$ is the chief executive node)}\\
    \hspace{2em} $\mathit{SEND}_j(ImP(z), y)$, {\scriptsize (//where $y$ are the total nodes in $out_j[]\setminus z$, whose states are active, and $z$ is the new $C_E$ that also belongs to $out_j[]$)} \\
    \hspace{2em} $\mathit{SEND}_j(ImPC(hold_j, b), z)$, {\scriptsize (//$b$ represents the total child nodes of $\mathit{CR}_j$, whose states are active, except $z$)}\\
    \hspace{2em} \textit{Three-wayHandshake}$()$,\\\\

    \hspace{1em} $\ [\!] parent_j \neq j \wedge \mathit{STATE}(parent_j)=\mathit{ACTIVE} \rightarrow$ {\scriptsize (//A case to show when $\mathit{CR}_j$ is not the chief executive node)}\\
    \hspace{2em} $\mathit{SEND}_j(ImPC(hold_j,k^{\prime}),parent_j))$, \textit{Three-wayHandshake}$()$,\\
    \hspace{2em} $\forall k^{\prime}: k^{\prime} \in out_j[] \wedge \mathit{STATE}(k^{\prime})= \mathit{ACTIVE} \rightarrow $ {\scriptsize ($k^{\prime}$ represents child nodes of $\mathit{CR}_j$ whose states are active)}\\
    \hspace{3em} $\mathit{SEND}_j(ImP(parent_j), k^{\prime})$, \\\\

    \hspace{1em} $\ [\!] parent_j \neq j \wedge \mathit{STATE}(parent_j)=\mathit{PASSIVE} \wedge (\exists ! k^{\prime}: k^{\prime} \in out_j[]\vee in_j[]) \wedge \mathit{STATE}(k^{\prime})= \mathit{ACTIVE} \rightarrow $ \\
    \hspace{2em}{\scriptsize (//A case to show when $\mathit{CR}_j$ is not $C_E$, $parent_j$ is passive, and there is at least one child node of $\mathit{CR}_j$)}\\
    \hspace{2em} $\mathit{SEND}_j(ImPC(hold_j,k^{\prime\prime}),k^{\prime})$, \textit{Three-wayHandshake}$()$, $\mathit{SEND}_j(ImP(k^{\prime}),k^{\prime\prime})$, {\scriptsize (// $k^{\prime\prime}\in out_j[]\setminus k^{\prime}$)}\\\\

    \hspace{1em} $\ [\!] parent_j \neq j \wedge \mathit{STATE}(parent_j)=\mathit{PASSIVE} \wedge (\nexists ! k^{\prime}: k^{\prime} \in out_j[]\vee in_j[]) \wedge \mathit{STATE}(k^{\prime})= \mathit{ACTIVE} \rightarrow $ \\
    \hspace{2em}{\scriptsize (//A case to show when $\mathit{CR}_j$ is not $C_E$, $parent_j$ is passive, and there is no child node of $\mathit{CR}_j$)}\\
    \hspace{2em} $\mathit{SEND}_j(ImPC(hold_j,0),z^{\prime})$, \textit{Three-wayHandshake}$()$, {\scriptsize (//where $z^{\prime}$ is a node that is executing the same computation as $\mathit{CR}_j$ did)}\\\\

    \hspace{1em} $parent_j = 0$, $hold_j:= 0, inj[i]:= \emptyset, outj[]:= \emptyset, session_j:= x$\\ \\

    \textbf{$A_5$.} \textbf{Reception of $\mathit{ImPC}(C,b)$ messages.} $\mathit{CR}_j$ receives $\mathit{ImPC}(C, b)$ from $\mathit{CR}_i \rightarrow$ \\
    \hspace{1em} $\ [\!] parent_i = j \wedge b= 0 \rightarrow hold_j:= hold_j + C+in_j[i], out_j[i]:= \emptyset, in_j[i]:= \emptyset$\\
    \hspace{1em} $\ [\!] parent_i = j \wedge b\neq 0 \rightarrow hold_j:= hold_j + C+in_j[i], out_j[i]:= \emptyset, out_j[] = out_j[] \cup out_i[],in_j[i]:= \emptyset$\\
    \hspace{1em} $\ [\!] parent_j= i \rightarrow parent_j:=j, hold_j:= hold_j + C+in_j[i], out_j[] = out_j[] \cup out_i[],in_j[i]:= \emptyset$\\
    \hspace{1em} $\ [\!] parent_j\neq i \wedge parent_i\neq j \wedge b=0 \rightarrow hold_j:= hold_j + C $ \\

    \hspace{1em} $\mathit{SEND}_j(AcK, i)$,\\
    \hspace{1em} Wait for a $t_e$ or an $AAcK$ from $\mathit{CR}_i$,\\
    \hspace{1em} \textbf{if} $ t_e  \wedge \neg AcK$ \textbf{then}\\
    \hspace{2em} send a special message to $C_E$ {\scriptsize (This special message avoids multiple credit surrender by $\mathit{CR}_i$ to different nodes.)}\\ \\

    \textbf{$A_6$.} \textbf{Reception of $\mathit{ImP}(p)$ messages.} $\mathit{CR}_j$ receives $\mathit{ImP}(p)$ from $\mathit{CR}_i \rightarrow$\\
    \hspace{1em} $\ [\!] parent_j= i \rightarrow parent_j:= p$\\
    \hspace{1em} $\ [\!] parent_j \neq i \rightarrow hold_j:= hold_j + in_j[i], in_j[i]:= \emptyset $\\\\

    {\bf Function \textit{Three-wayHandshake}}$() \{$  \\
    \hspace{1em} Wait for a $t_e$ or $AcK$,\\
    \hspace{1em} \textbf{if} $t_e \wedge \neg AcK$  \textbf{then} $\mathit{SEND}_j(ImPC(hold_j, b), z^{\prime})$, {\scriptsize (//where $z^{\prime}$ is the new $C_E$)}\\
    \hspace{1em} \textbf{else} $\mathit{SEND}_j(AAcK, z) \}$ \\\\

    \hline
    \end{tabular}
\BB
\caption{The credit distribution and aggregation phases.}
\BBB
\label{tab:Actions of credit diffusion-aggregation}
\end{center}
\end{table}

\section{The \textit{T-CRAN} Protocol as Guarded Actions}
\label{section:T-CRAN in Guard-Action Paradigm}
We specify our \textit{T-CRAN} protocol as guarded actions. A guarded action is written as: $\langle Guard\rangle \rightarrow \langle Action(s) \rangle$. A guard (or predicate) of actions (or rules) is a Boolean expression, and if a guard is true, then all the actions, corresponding to that guard, are executed in an atomic manner. At some point of time more than one guard may be true. The \emph{T-CRAN} protocol as guarded actions is given in Tables~\ref{tab:Actions of credit diffusion-aggregation},~\ref{tab:Actions of disconnection and rejoining of the affected nodes}, and~\ref{tab:Actions of termination announcement}. For the sake of simplicity, we assume that all the guards are related to an identical session of a computation.

\subsection{The credit distribution and aggregation phases}
\label{subsec:Actions of credit diffusion-aggregation}
The following actions $A_1$ through $A_6$ define the credit distribution and aggregation phases, refer to Table~\ref{tab:Actions of credit diffusion-aggregation}.
\begin{description}[noitemsep]
  \item[$A_1$] \textbf{Computation and protocol initiation.} $A_1$ is executed on reception of a message $M$ from outside world, and then, a node initiates a computation and the \textit{T-CRAN} protocol.

  \item[$A_2$] \textbf{Distribution of credits.} In $A_2$, $\mathit{CR}_i$ distributes the computation among its $q$ ($q\leq N$) neighboring nodes using $q$ different \textit{COMputation message}s. Note that the \textit{T-CRAN} protocol executes concurrently over the computation.

  \item[$A_3$] \textbf{Reception of \textit{COMputation message}s.} The first guard shows that $\mathit{CR}_j$, whose state is passive, receives a \textit{COMputation message} for the first time from $\mathit{CR}_i$. Hence, $\mathit{CR}_i$ becomes the parent node of $\mathit{CR}_j$, and $\mathit{CR}_j$ keeps credit, $C_j$, in variable $hold_j$. The second guard shows that $\mathit{CR}_j$, whose state is active, receives a \textit{COMputation message} from $\mathit{CR}_i$. Hence, $\mathit{CR}_j$ does not assign $\mathit{CR}_i$ as its parent, and $\mathit{CR}_j$ keeps credit, $C_j$, in $in_j[i]$.

  \item[$A_4$] \textbf{Credit surrender.} $A_4$ is executed when a node finishes its computation and consequently becomes passive. The first statement shows that $\mathit{CR}_j$ sends \textit{I am Passive with Credit message}s to all the $k$ nodes whose states are active and they had sent some credits to $\mathit{CR}_j$.

      The first guard becomes true when $C_E$ becomes passive. $C_E$ sends \textit{I am Passive message}s ($\mathit{ImP}(p)$) to all its child nodes with the information of the new $C_E$. In addition, the old $C_E$ sends an \textit{I am Passive with Credit message} to the new $C_E$ and executes the function $\textit{Three-wayHandshake()}$ to ensure delivery of its credit ($hold_{C_{E}}$) to the new $C_E$.

      The second guard becomes true when $\mathit{CR}_j$ becomes passive and $\mathit{CR}_j$ is not the chief executive node. $\mathit{CR}_j$ informs all its child nodes about the new parent node using $Imp(p)$ messages. Also, $\mathit{CR}_j$ sends an \textit{I am Passive with Credit message} to its parent node with its credit ($hold_j$) and executes the function $\textit{Three-wayHandshake()}$ to ensure delivery of its credit.

      The third guard becomes true when $\mathit{CR}_j$ becomes passive and its parent node is also passive. However, there is at least a node $\mathit{CR}_{k^{\prime}}$ in $out_j[]$ or $in_j[]$ whose state is active. $\mathit{CR}_j$ sends its credit to $\mathit{CR}_{k^{\prime}}$ using an \textit{I am Passive with Credit message} and executes the function $\textit{Three-wayHandshake()}$. Also, $\mathit{CR}_j$ sends $\mathit{ImP}(p)$ messages to all the nodes whose states are active and in $out_j[]$.

      The forth guard becomes true when $\mathit{CR}_j$ becomes passive, its parent node is also passive, and there is no node in $out_j[]$ or $in_j[]$ whose state is active. $\mathit{CR}_j$ sends its credit to $\mathit{CR}_{z^{\prime}}$ using an \textit{I am Passive with Credit message} and executes the function $\textit{Three-wayHandshake()}$. When $\mathit{CR}_{z^{\prime}}$ is selected, a priority is given to a node that is executing the same computation as $\mathit{CR}_j$.

  \item[$A_5$] \textbf{Reception of $\mathit{ImPC}(C,b)$ messages.} The reception of \textit{I am Passive with Credit message}s at $\mathit{CR}_j$ is shown in $A_5$. The first guard becomes true when $\mathit{CR}_j$ receives $\mathit{ImPC}(C,b)$ messages from its child nodes that do not have any child node. The second guard becomes true when $\mathit{CR}_j$ receives $\mathit{ImPC}(C,b)$ messages from its child nodes that have some child nodes. The third guard becomes true when $\mathit{CR}_j$ receives an $\mathit{ImPC}(C,b)$ from its parent node. The forth guard becomes true when $\mathit{CR}_j$ receives $\mathit{ImPC}(C,b)$ messages from any node that is executing the same computation as $\mathit{CR}_j$ is executing. In all the cases, $\mathit{CR}_j$ also sends an $AcK$ to $\mathit{CR}_i$, and if $\mathit{CR}_j$ does not receive an $AAcK$ from $\mathit{CR}_i$, it sends a special message to $C_E$ that is used to balance the credit in the system. (The working of this special message is shown in Figure~\ref{fig:credit_does_not_exceed2}).

  \item[$A_6$] \textbf{Reception of $\mathit{ImP}(p)$ messages.} The reception of \textit{I am Passive message}s at $\mathit{CR}_j$ is given in $A_6$. The first guard is related to the arrival of an $\mathit{ImP}(p)$ at a child node $\mathit{CR}_j$ from its parent node, and the second guard is related to the arrival of an $\mathit{ImP}(p)$ from any node other than its parent node.
\end{description}

\subsection{The appearance and disappearance of primary users}
\label{subsec:Actions of disconnection and rejoining of affected nodes}
The following actions $B_1$ through $B_4$ are related to the appearance and disappearance of primary users on channels, refer to Table~\ref{tab:Actions of disconnection and rejoining of the affected nodes}.
\begin{description}[noitemsep]
  \item[$B_1$] \textbf{Appearance of a PU.} $B_1$ becomes true when a PU appears on a channel. The first guard becomes true when $\mathit{CR}_i$ has some channels in its $\mathit{LCS}$ and tunes to one of them. The second guard becomes true when $\mathit{CR}_i$'s $\mathit{LCS}$ is empty; hence, $\mathit{CR}_i$ becomes an affected node. In addition, all the $k$ neighbors, whose states are active, of the affected node $\mathit{CR}_i$ send $\mathit{PaN}$ messages to $C_E$.

  \item[$B_2$] \textbf{Reception of $\mathit{PaN}$ messages.} When $C_E$ receives $\mathit{PaN}$ messages, it places the affected nodes and their credits, which are received with $\mathit{PaN}$ messages, in the respective data structures.

  \item[$B_3$] \textbf{Disappearance of a PU.} When a PU leaves the channel, all the affected nodes become non-affected nodes, and they send $\mathit{NaP}$ messages to $C_E$ and all its neighboring nodes, whose states are active.

  \item[$B_4$] \textbf{Reception of $\mathit{NaP}$ messages from $\mathit{CR}_j$.} The reception of a $\mathit{NaP}$ at $C_E$, from any node, notifies the absence of PU(s). The first guard becomes true when $C_E$ receives $\mathit{NaP}$ messages for the current computation. $C_E$ removes the sender of the $\mathit{NaP}$ (that is $\mathit{CR}_j$) from its $PU_{\mathit{affected}}[]$ and informs $\mathit{CR}_j$ about the ongoing computation. The second guard becomes true when the neighboring nodes of $\mathit{CR}_i$ receive $\mathit{NaP}$ messages. All the receiver nodes send a special message $m^{\prime}$ to $C_E$ that is a request to receive back the credit that they had sent earlier when $\mathit{CR}_i$ was an affected node.
\end{description}

\begin{table}
\begin{center}
\small
    \begin{tabular}{| p{17cm} |}
    \hline
    \textbf{Notations:} $\mathit{SEND}_i(m, j)$: $\mathit{CR}_i$ sends a message $m$ to $\mathit{CR}_j$, $\mathit{STATE}(i)$: current state of $\mathit{CR}_i$, \textit{i}.\textit{e}., either active or passive, $\mathit{Neighbor}_i[]$: neighboring nodes of $\mathit{CR}_i$ that are executing the same computation as $\mathit{CR}_i$. All the data structures have usual meanings (see Table~\ref{table:datastructure} for details of the data structures).         \\ \\

    \textbf{$B_1$.} \textbf{Appearance of a PU.} A primary user appears on channels $\rightarrow$\\
    \hspace{1em} $\ [\!] LCS_i\neq \emptyset \rightarrow \mathit{CR}_i$ leaves the channel and tune to another available channel\\
    \hspace{1em} $\ [\!] LCS_i= \emptyset \rightarrow \mathit{CR}_i$ becomes an affected nodes and is not allowed to send and receive messages,\\
    \hspace{2em} $\forall k\in \mathit{Neighbor}_i[] \wedge \mathit{STATE}(k) = \mathit{ACTIVE} \rightarrow \mathit{SEND}_k(PaN, C_E)$, where a $\mathit{PaN}$ holds $\langle \mathit{CR}_i, in_k[i], out_k[i]\rangle$,\\
    \hspace{2em} $in_k[i]=\emptyset, out_k[i]=\emptyset$\\\\

    \textbf{$B_2$.} \textbf{Reception of $\mathit{PaN}$ messages.} $C_E$ receives a $\mathit{PaN}$ that holds $\langle \mathit{CR}_i, in_k[i], out_k[i]\rangle \rightarrow$ \\
    \hspace{1em} $PU_{\mathit{affected}_{C_{E}}}[i]:= \mathit{CR}_i, C\_PU_{\mathit{affected}_{C_{E}}}[i]:= \langle in_k[i], out_k[i]\rangle$ \\\\

    \textbf{$B_3$.} \textbf{Disappearance of a PU.} A primary users disappear from channels $\rightarrow$\\
    \hspace{1em} $ LCS_i\neq \emptyset$, $\mathit{SEND}_i(NaP, C_E)$, $\forall k\in \mathit{Neighbor}_i[] \wedge \mathit{STATE}(k) = \mathit{ACTIVE} \rightarrow \mathit{SEND}_i(NaP, k)$ \\\\

    \textbf{$B_4$.} \textbf{Reception of $\mathit{NaP}$ messages from $\mathit{CR}_j$.} \\
    \hspace{1em} $\ [\!] i = C_E \wedge session_{C_{E}} = x \wedge \mathit{STATE}(C_E)=\mathit{ACTIVE} \rightarrow$\\
    \hspace{2em} $PU_{\mathit{affected}_{C_{E}}}[j]:= \emptyset, C\_PU_{\mathit{affected}_{C_{E}}}[j]:= \emptyset, \mathit{SEND}_{C_{E}}(m, \mathit{CR}_j)$\\

    \hspace{1em} $\ [\!] i \neq C_E \wedge j \in \mathit{Neighbor}_i[] \wedge session_i = x \wedge \mathit{STATE}(i)=\mathit{ACTIVE} \rightarrow \mathit{SEND}_i(m^{\prime}, C_E)$\\

    \hline
    \end{tabular}
\BB
\caption{The appearance and disappearance of primary users.}
\BBB
\label{tab:Actions of disconnection and rejoining of the affected nodes}
\end{center}
\end{table}

\begin{table}
\small
\begin{center}
    \begin{tabular}{| p{17cm} |}
    \hline
    \textbf{$C_1$.} \textbf{Global weak termination.} \\
    \hspace{1em} $out_{C_{E}}[]= \emptyset \wedge in_{C_{E}}[]= \emptyset \wedge (hold_{C_{E}} + C\_PU_{\mathit{affected}_{C_{E}}}[] = C) \wedge \mathit{STATE}(C_E) = \mathit{PASSIVE}\rightarrow$ \\
    \hspace{2em} Announce global weak termination\\ \\

   \textbf{$C_2$.} \textbf{Global strong termination.} \\
   \hspace{1em} $out_{C_{E}}[]= \emptyset \wedge in_{C_{E}}[]= \emptyset \wedge C\_PU_{\mathit{affected}_{C_{E}}}[] = \emptyset \wedge hold_{C_{E}} = C \wedge \mathit{STATE}(C_E) = \mathit{PASSIVE}\rightarrow$ \\
   \hspace{2em} Announce global strong termination\\
    \hline
    \end{tabular}
\BB
\caption{The termination announcement.}
\BBB\BBB
\label{tab:Actions of termination announcement}
\end{center}
\end{table}

\subsection{The termination announcement}
\label{subsec:Actions of termination announcement}
The actions $C_1$ and $C_2$ represent termination detection in the networks, refer to Table~\ref{tab:Actions of termination announcement}.
\begin{description}[noitemsep]
  \item[$C_1$] \textbf{Global weak termination.} $C_1$ announces the global weak termination when only a single node (that is $C_E$) has the credit $C$, which was distributed at the time of computation initiation, in its $hold_{C_{E}}$ and $C\_PU_{\mathit{affected}}[]$.

  \item[$C_2$] \textbf{Global strong termination.} $C_2$ announces the global strong termination when $C_E$ contains the total credit $C$, which was distributed at the time of computation initiation, in its $hold_{C_{E}}$, no credit in $C\_PU_{\mathit{affected}}[]$, and $C_E$ has become passive.
\end{description}

\begin{wraptable}{r}{6cm}
\BBB
\begin{center}
    \begin{tabular}{|l|l|}
    \hline
    {\scriptsize Nodes} & {\scriptsize Local Channel Sets ($\mathit{LCS}$)} \\ \hline
    {\scriptsize 1}     & {\scriptsize $2,3,\textbf{5}$}                  \\ \hline
    {\scriptsize 2}     & {\scriptsize $3,\textbf{5},6,9$}               \\ \hline
    {\scriptsize 3}     & {\scriptsize $\textbf{5}$}                        \\ \hline
    {\scriptsize 4}     & {\scriptsize $\textbf{5}$}                        \\ \hline
    {\scriptsize 5}     & {\scriptsize $\textbf{5},7,9$}                  \\ \hline
    {\scriptsize 6}     & {\scriptsize $\textbf{5},9$}                      \\ \hline
    \end{tabular}
\caption{Local channel sets of all the nodes.}
\label{tab:LCS_of_nodes}
\end{center}
\BBB
\end{wraptable}

\section{The Working of the \textit{T-CRAN} protocol}
\label{section:The Working of the T-CRAN}

In Figures~\ref{fig:exe_withotu_PU} and~\ref{fig:exe_with_PU}, a sample execution of the \textit{T-CRAN} protocol is represented in the absence and presence of a primary user, respectively. The local channel set for every node is given in Table~\ref{tab:LCS_of_nodes}, where boldface characters show the currently tuned channel at the respective nodes. Actions $A_1$ to $A_6$, actions $B_1$ to $B_4$, and actions $C_1$, $C_2$ are given in Tables~\ref{tab:Actions of credit diffusion-aggregation},~\ref{tab:Actions of disconnection and rejoining of the affected nodes}, and~\ref{tab:Actions of termination announcement}. Also, all the nodes are in the transmission range of each other.

In Figure~\ref{fig:exe1}, cognitive radio node 1 initiates an assigned computation and the \textit{T-CRAN} protocol, hence, behaves as the chief executive node, $C_E$. Node 1 sends a \textit{COMputation message} ($\mathit{COM}(0.9)$) to node 2, and node 1 holds a credit value 0.1. Further, nodes 2, 3, 4, and 5 also distribute the computation with some credit values and initialize their respective data structures ($A_1$, $A_2$, and first guard of $A_3$). Note that, $in_i[]$ array of all the nodes is empty initially, and the sum of credits at nodes 1-6 equals 1.

In Figure~\ref{fig:exe2}, node 5 sends a \textit{COMputation message} ($\mathit{COM}(0.1)$) to node 3. The reception of $\mathit{COM}(0.1)$ initializes $in_3[5]$ array variable at node 3 and corresponding entry in $out_5[3]$ (second guard of $A_3$). Also, node 6 sends two \textit{COMputation message}s to nodes 4 and 3.
\begin{figure}
\centering
    \begin{minipage}[b]{.5\textwidth}
        \centering
          \includegraphics[width=60mm, height=35mm]{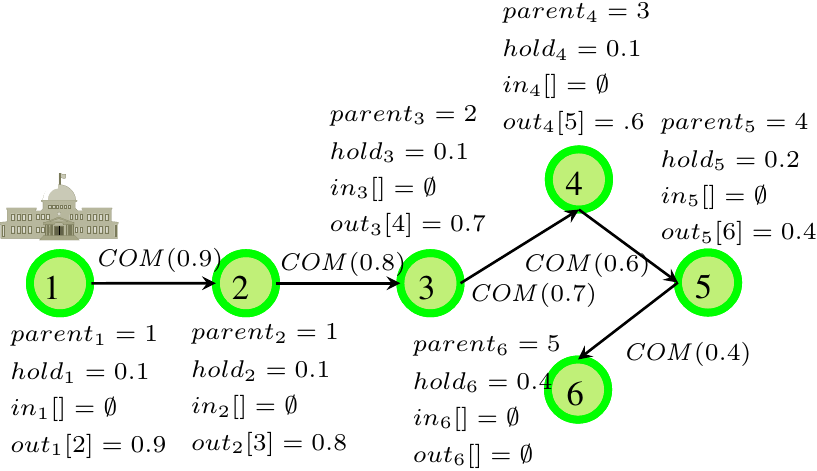}
       \subcaption{}
       \label{fig:exe1}
    \end{minipage}
    \begin{minipage}[b]{.49\textwidth}
        \centering
         \includegraphics[width=60mm, height=35mm]{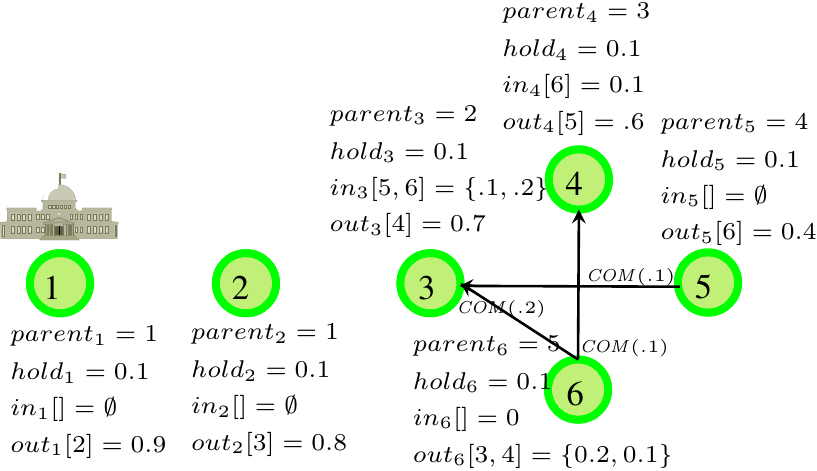}
        \subcaption{}
        \label{fig:exe2}
    \end{minipage}
    \begin{minipage}[b]{.5\textwidth}
        \centering
          \includegraphics[width=60mm, height=35mm]{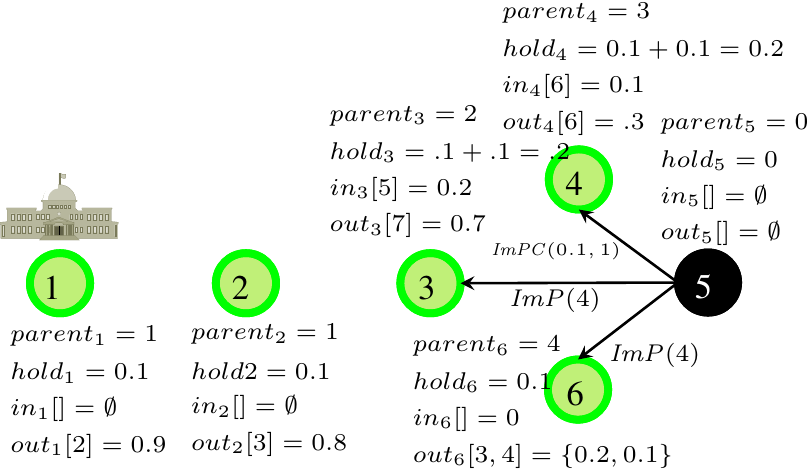}
       \subcaption{}
       \label{fig:exe3}
    \end{minipage}
    \begin{minipage}[b]{.49\textwidth}
        \centering
         \includegraphics[width=60mm, height=35mm]{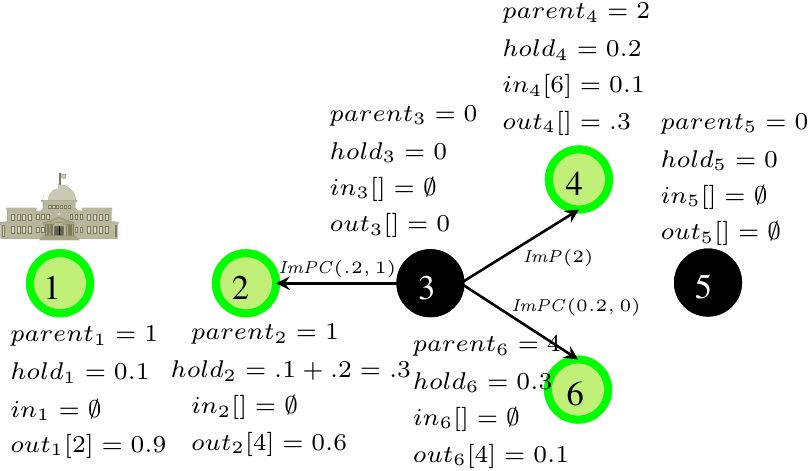}
        \subcaption{}
        \label{fig:exe4}
    \end{minipage}
    \begin{minipage}[b]{.5\textwidth}
        \centering
          \includegraphics[width=60mm, height=35mm]{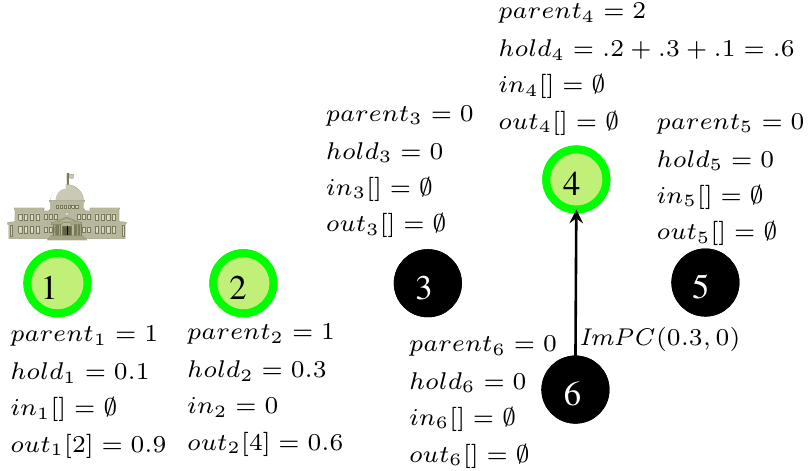}
       \subcaption{}
       \label{fig:exe5}
    \end{minipage}
    \begin{minipage}[b]{.49\textwidth}
        \centering
         \includegraphics[width=60mm, height=35mm]{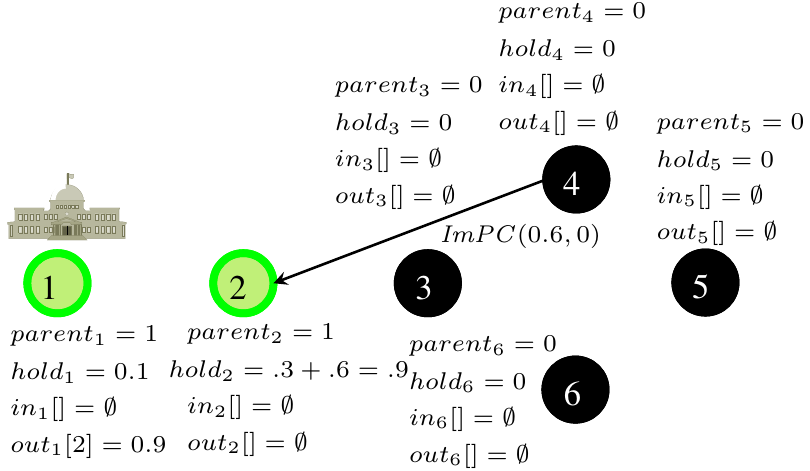}
        \subcaption{}
        \label{fig:exe6}
    \end{minipage}
    \begin{minipage}[b]{.5\textwidth}
        \centering
          \includegraphics[width=60mm, height=35mm]{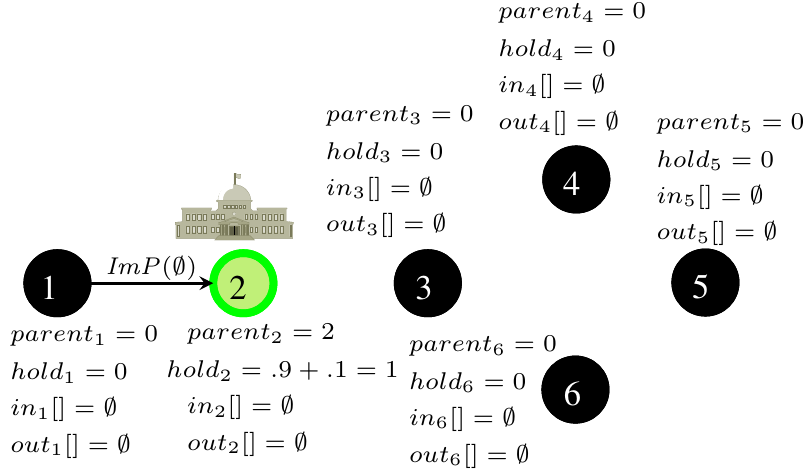}
       \subcaption{}
       \label{fig:exe7}
    \end{minipage}
    \begin{minipage}[b]{.49\textwidth}
        \centering
         \includegraphics[width=45mm, height=25mm]{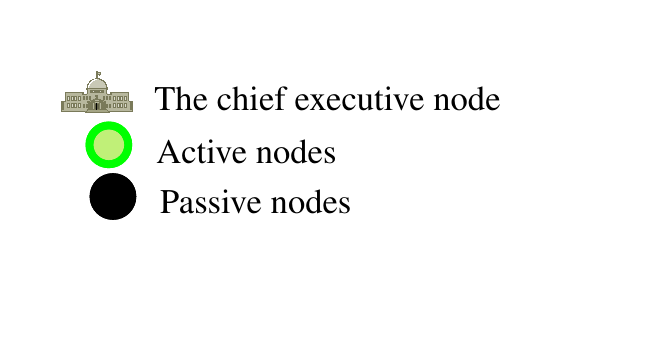}
        \subcaption{Notations}
        \label{fig:exe_notations_without_PU}
    \end{minipage}

\B
\caption{The Working of the \textit{T-CRAN} protocol in the absence of PUs.}
\BBB
\label{fig:exe_withotu_PU}
\end{figure}

In Figure~\ref{fig:exe3}, node 5 becomes passive and surrenders its credit (contained in variable $hold_5$) to its parent node 4 via an \textit{I am Passive with Credit message}, $\mathit{ImPC}(0.1,1)$ (second guard of $A_4$). In addition, node 5 sends two \textit{I am Passive message}s ($\mathit{ImP}(4)$) to its child nodes 3 and 6. These \textit{I am Passive message}s hold information of the new parent node 4. The reception of $\mathit{ImPC}(0.1,1)$ at node 4 triggers second guard of $A_5$. The reception of $\mathit{ImP}(4)$ triggers first guard of $A_6$ at node 6 so that node 6 assigns node 4 as its parent node, and second guard of $A_6$ at node 3 so that node 3 holds 0.2 credits. Note that we are not showing the three way handshake for clarity of figures, interested readers can look ahead to Figure~\ref{fig:credit_does_not_exceed2} to see the working of the three-way handshake.

In Figure~\ref{fig:exe4}, node 3 becomes passive and sends: (\textit{i}) an $\mathit{ImPC}(0.2,0)$ to node 6 (the first line of $A_4$), and node 6 holds back 0.3 credits (first guard of $A_5$), (\textit{ii}) an $\mathit{ImPC}(0.1,1)$ to its parent node 2 (second guard of $A_4$), and node 2 holds 0.3 credits now (second guard of $A_5$), (\textit{iii}) an $\mathit{ImP}(2)$ to node 4 (second guard of $A_4$), and node 4 assigns node 2 as its parent node (first guard of $A_5$).

In Figure~\ref{fig:exe5}, node 6 becomes passive and surrenders its credit to node 4 (second guard of $A_4$). Node 4 holds 0.6 credits now (first guard of $A_5$). In Figure~\ref{fig:exe6}, node 4 becomes passive and surrenders its credit to node 2 via an $\mathit{ImPC}(0.6,0)$ (second guard of $A_4$).

In Figure~\ref{fig:exe7}, the chief executive node 1 becomes passive and surrenders its credit to node 2 (first guard of $A_4$), and node 2 now becomes the new chief executive node, holds credit 1 (third guard of $A_5$). Once the node 2 finishes its computation, node 2 announces the global strong termination (according to $C_2$).

\begin{figure}
\centering
    \begin{minipage}[b]{.5\textwidth}
        \centering
          \includegraphics[width=60mm, height=35mm]{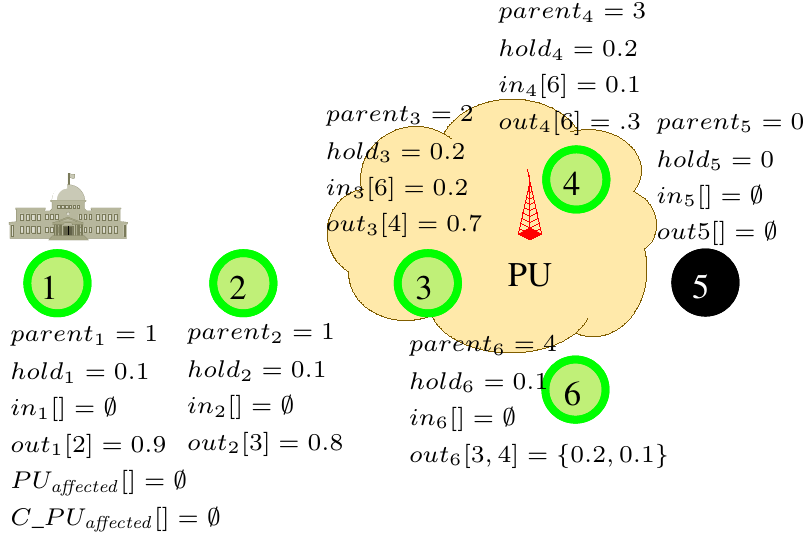}
       \subcaption{}
       \label{fig:exe8}
    \end{minipage}
    \begin{minipage}[b]{.49\textwidth}
        \centering
         \includegraphics[width=60mm, height=35mm]{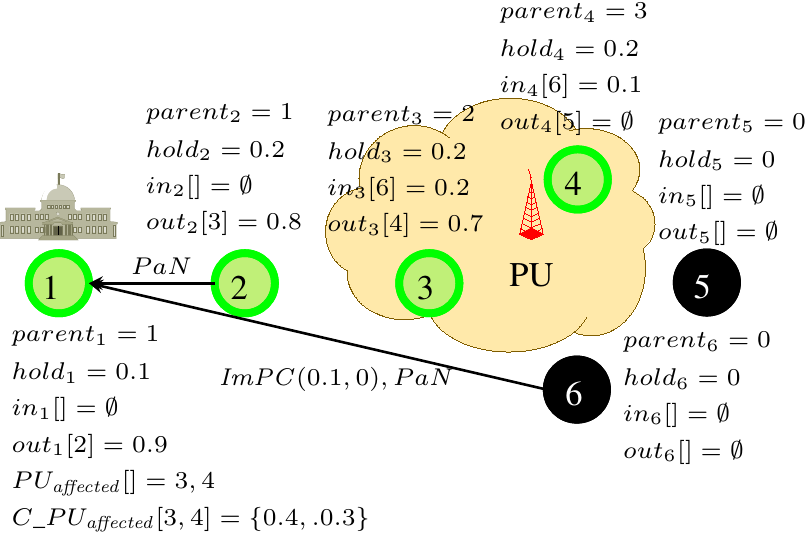}
        \subcaption{}
        \label{fig:exe9}
    \end{minipage}
    \begin{minipage}[b]{.5\textwidth}
        \centering
          \includegraphics[width=60mm, height=35mm]{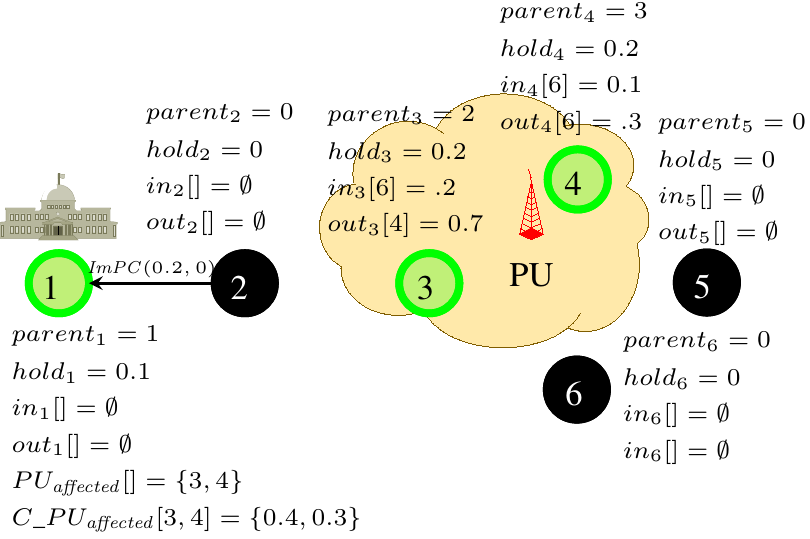}
       \subcaption{}
       \label{fig:exe10}
    \end{minipage}
    \begin{minipage}[b]{.49\textwidth}
        \centering
         \includegraphics[width=60mm, height=35mm]{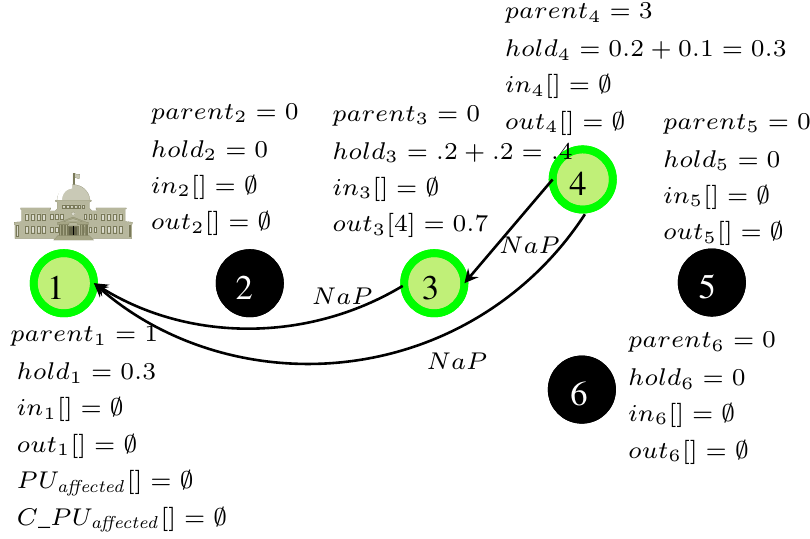}
        \subcaption{}
        \label{fig:exe11}
    \end{minipage}

     \begin{minipage}[b]{.99\textwidth}
        \centering
          \includegraphics[width=120mm]{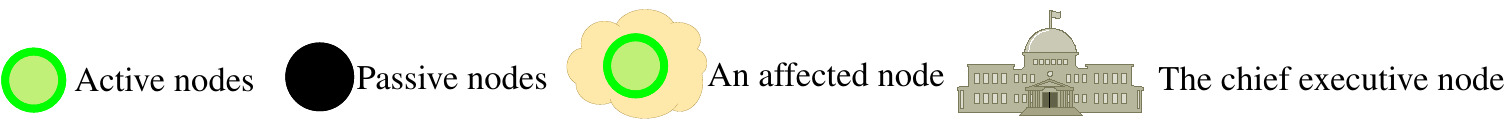}
       \subcaption{Notations}
       \label{fig:notations_with_PU}
    \end{minipage}
\B
\caption{The Working of the \textit{T-CRAN} protocol in the presence of PUs.}
\BBB
\label{fig:exe_with_PU}
\end{figure}
Figure~\ref{fig:exe_with_PU} shows the presence of a PU. Figure~\ref{fig:exe3}, where there is no PU, turns to Figure~\ref{fig:exe8} in the presence of a PU, where nodes 3 and 4 are affected nodes (second guard of $B_1$). In Figure~\ref{fig:exe9}, nodes 6 and 2 detect affected nodes 4 and 3, and inform node 1 using $\mathit{PaN}$ messages (second guard of $B_1$), and node 1 places information of the affected nodes 3, 4 in the respective data structures (according to $B_2$). In addition, node 6 surrenders its credit to node 2 (forth guard of $A_4$).

In Figure~\ref{fig:exe10}, node 2 becomes passive and surrenders its credit to node 1 (second guard of $A_2$). Now, node 1 holds the credit value 1, which was distributed at the time of computation's initiation, in $hold_1$ and $C\_PU_{\mathit{affected}}[]$, and this condition is sufficient to announce the global weak termination after a timeout (according to $C_1$). In Figure~\ref{fig:exe11}, the primary user disappears, and nodes 3, 4 inform node 1 ($C_E$) and ask about the ongoing computation (according to $B_3$). Node 1 informs them and deletes their entry from $PU_{\mathit{affected}}$ and $C\_PU_{\mathit{affected}}[]$, and then, the sum of credits at node 1, 3, and 4 equal to 1.

\section{Conclusion}
\label{section:conclusion}
A termination detection protocol, \textit{T-CRAN}, for an asynchronous multi-hop cognitive radio networks is presented. The \textit{T-CRAN} protocol is capable enough to work on heterogeneous channels, and it can also handle multiple computations simultaneously. The \textit{T-CRAN} protocol is based on credit distribution and aggregation approach. The proposed protocol uses a new kind of logical structure, called the \textit{virtual tree-like structure}. In the \textit{virtual tree-like structure}, a node may surrender its credit to any node (not necessarily to its parent node) that is executing the identical computation. This credit surrender approach significantly reduces the waiting time to announce termination. Further, it is not mandatory for the initiator of the protocol (\textit{i}.\textit{e}., the first root node of the \emph{virtual tree-like structure}) to stay involved until the termination of the computation. Hence, the protocol may witness different root nodes at different time instants, during the course of termination announcement.

The proposed protocol can also be implemented in dynamic networks, \textit{e}.\textit{g}., cellular, mobile ad hoc networks, and vehicular ad hoc networks. The proposed \textit{virtual tree-like structure} is also able to decrease the waiting time to announce termination in dynamic networks, which is a desirable requirement in dynamic networks, due to its flexible credit surrender approach. Moreover, the \textit{virtual tree-like structure} can substitute the conventional tree structures in various distributed computations, \textit{e}.\textit{g}., snapshot, global-state, leader election, message ordering, and group communication.


\appendix

\section{Complexity Analysis}
\label{section:Complexity Analysis}
We analyze our protocol in terms of \textit{message complexity} and \textit{time complexity}. Message complexity is defined in terms of the total number of control messages that are used in our protocol. Time complexity is defined in terms of the time elapsed between the initiation of the protocol and the announcement of termination. Notations used to analyze our protocol are given in Table~\ref{table:notations_analysis}.

\subsection{Message complexity}
\label{subsec:Message Complexity}
We are using eight types of messages in our protocol (see Table~\ref{tab:messages}). We analyze each message separately, except the \textit{Termination Message} ($\mathit{TM}$). The $\mathit{TM}$ is sent by $C_E$ (when $C_E$ receives back credit $C$ that was used at the time of credit distribution) to all the nodes of the interaction graph to declare termination of the computation. Since $\mathit{TM}$ can be broadcasted to all the nodes in a unit time, we ignore to analyze this message. The message complexity of each message is also given in Table~\ref{table:message_complexity}.

\begin{itemize}[noitemsep,nolistsep]
  \item \textit{COMputation message} ($\mathit{COM}(C)$): A node may send $\mathit{COM}(C)$ messages to all its neighboring nodes to distribute the computation. Since there are $N$ nodes in \textit{CRN} and the maximum allowable degree of a node is $\triangle$, $\BigO(N\times\triangle)$ \textit{COMputation message}s can be exchanged in the protocol.

  \item \textit{I am Passive with Credit message} ($\mathit{ImPC}(C,b)$): A node sends an $\mathit{ImPC}(C,b)$ to its parent node, if the parent node is active, and to all the neighboring nodes, whose states are active and had sent credits to the node previously. Since at most $N_{\mathit{neighbor}}$ nodes of a node may execute the same protocol and at most $N_{leave}$ nodes may leave the network, the message complexity of $\mathit{ImPC}(C,b)$ is $\BigO(N_{\mathit{neighbor}}\times N_{leave})$.

  \item \textit{I am Passive message} ($\mathit{ImP}(p)$): A node sends $\mathit{ImP}(p)$ messages to all its child nodes, whose states are active. Since a single node sends at most $(N_{\mathit{neighbor}}-1)$ $\mathit{ImP}(p)$ messages and at most $N_{leave}$ nodes may leave the network, the message complexity of $\mathit{ImP}(p)$ is $\BigO((N_{\mathit{neighbor}}-1)\times N_{leave})$.

  \item \textit{AcKnowledgement message} ($AcK$) and \textit{Acknowledgement of AcK message} ($AAcK$): An $AcK$ and an $AAcK$ is generated in response to an $\mathit{ImPC}(C,b)$; hence, $\BigO(N_{\mathit{neighbor}}\times N_{leave})$ $AcK$ and $AAcK$ messages can be generated in the protocol.
\end{itemize}

\medskip

We also present the total number of non-control messages, as follows:
\begin{itemize}[noitemsep,nolistsep]
  \item \textit{Primary user affected Nodes message} ($\mathit{PaN}$): All the neighboring nodes of an affected node, whose states are active, send $\mathit{PaN}$ messages to $C_E$. Since at most $N_{\mathit{neighbor}}$ nodes of at most $N_{\mathit{affected}}$ nodes may send $\mathit{PaN}$ messages, the message complexity of $\mathit{PaN}$ is $\BigO(N_{\mathit{neighbor}}\times N_{\mathit{affected}})$.

  \item \textit{Nodes released by Primary user message} ($\mathit{NaP}$): After recovery, the affected node sends $\mathit{PaN}$ messages to its neighboring node and $C_E$. Since at most $N_{\mathit{neighbor}}$ and $C_E$ receive $\mathit{PaN}$ messages from $N_{\mathit{affected}}$ nodes, the message complexity of $\mathit{NaP}$ is $\BigO((N_{\mathit{neighbor}}+1)\times N_{\mathit{affected}})$.
\end{itemize}

\begin{table}
\begin{center}
\begin{spacing}{0.85}
    \begin{tabular}{|l|l|}
    \hline
    {\scriptsize Notations}       & {\scriptsize Description   }                                                                     \\ \hline
    {\scriptsize $N$}             & {\scriptsize The total number of the cognitive radio nodes in the network.}                           \\ \hline
    {\scriptsize $\triangle$}     & {\scriptsize Maximum degree of a nodes in the network.}                                           \\ \hline
    {\scriptsize $\mathit{Height}$}        & {\scriptsize Maximum height of the \textit{virtual tree-like structure}.} \\\hline
    {\scriptsize $N_{leave}$}     & {\scriptsize The total number of nodes that can leave the network during the protocol execution.} \\ \hline
    {\scriptsize $N_{\mathit{affected}}$}  & {\scriptsize Maximum number of nodes that are affected during the protocol execution.}            \\ \hline
    {\scriptsize $N_{\mathit{neighbor}}$}  & {\scriptsize Maximum number of neighboring nodes that are executing the same computation.}         \\ \hline
    \end{tabular}
\B
\caption{Notations used in the complexity analysis of the \textit{T-CRAN} protocol.}
\label{table:notations_analysis}
\BBB\BBB
\end{spacing}
\end{center}
\end{table}

\begin{table}
\begin{center}
\begin{spacing}{0.85}
    \begin{tabular}{|l|l|}
    \hline
    {\scriptsize Messages}                          & {\scriptsize Complexity }\\ \hline
    {\scriptsize \textit{COMputation message}   }           &   {\scriptsize $\BigO(N\times\triangle)$  }     \\ \hline
    {\scriptsize \textit{I am Passive with Credit message}}  &  {\scriptsize $\BigO(N_{\mathit{neighbor}}\times N_{leave})$}         \\ \hline
    {\scriptsize \textit{I am Passive message}             } &  {\scriptsize $\BigO((N_{\mathit{neighbor}}-1)\times N_{leave})$ }        \\ \hline
    {\scriptsize \textit{Primary user affected Nodes message}}    &  {\scriptsize $\BigO(N_{\mathit{neighbor}}\times N_{\mathit{affected}})$ }        \\ \hline
    {\scriptsize \textit{Nodes released by Primary user message}} &   {\scriptsize $\BigO((N_{\mathit{neighbor}}+1)\times N_{\mathit{affected}})$}       \\ \hline
    {\scriptsize \textit{AcKnowledgement message}         }  &  {\scriptsize $\BigO(N_{\mathit{neighbor}}\times N_{leave} )$         }\\ \hline
    {\scriptsize \textit{Acknowledgement of AcK message}   }  & {\scriptsize $\BigO(N_{\mathit{neighbor}}\times N_{leave})$         } \\ \hline
    \end{tabular}
\B
\caption{Message complexity of the \textit{T-CRAN} protocol.}
\label{table:message_complexity}
\BBB\BBB
\end{spacing}
\end{center}
\end{table}

\subsection{Time complexity}
\label{subsec:Time Complexity}
There are three types of nodes in the network: (\textit{i}) nodes whose $out[]=\emptyset$, (\textit{ii}) $C_E$, and (\textit{iii}) node whose $out[]\neq \emptyset$ or $in[]\neq \emptyset$ and they are not the chief executive node. The nodes with $out[] = \emptyset$ may leave the network when they finish their computation by sending $\mathit{ImPC}(C,b)$ messages to at most $N_{\mathit{neighbor}}$ nodes. Similarly, $C_E$ may also leave by sending $\mathit{ImPC}(C,b)$ or $\mathit{ImP}(p)$ messages to at most $N_{\mathit{neighbor}}$ nodes. Also, the node other than $C_E$ that has $out[]\neq \emptyset$ or $in[]\neq \emptyset$ exchanges at most $N_{\mathit{neighbor}}$ messages before leaving the network. We assume that all the messages are delivered in a unit time. Hence, in the failure-free network, all the nodes of the network take $\BigO(Height)$ time to leave the network that results in global strong termination declaration. However, the presence of PUs increases termination latency. In such scenarios, the declaration of global weak termination would be delayed according to the value of timeout.

\section{Correctness Proof}
\label{section:correctness_proof}
We first provide the system invariants; afterward, we prove the safety and liveness properties of the \textit{T-CRAN} protocol. We also prove an impossibility result that the appearance of a primary user on a single channel may defy termination forever.

\subsection{System invariants}
\label{subsec:System Invariants}
\begin{inv}
\label{inv:active_passive_rule}
Let, $\mathit{STATE}(i)$ represents the state of $\mathit{CR}_i$, which may be active or passive. For $\mathit{CR}_i$, $hold_i=0$ indicates passive state of $\mathit{CR}_i$ and vice versa. Also, $hold_i\neq 0$ indicates active state of $\mathit{CR}_i$ and vice versa.
\begin{equation*}\label{}
\forall i: hold_i = 0 \Leftrightarrow \mathit{STATE}(i)=\mathit{PASSIVE}, \: \forall i: hold_i \neq 0 \Leftrightarrow \mathit{STATE}(i)=\mathit{ACTIVE}
\end{equation*}
\end{inv}


\begin{inv}
\label{inv:credit_never_exceeds_C}
In \textit{CRN}, the sum of credits at the nodes and credits associated with in-transit messages must be $C$.
\begin{equation*}\label{}
\forall i, j \in v: hold_i + hold_j + in_i[] + in_j[] + \mathit{SEND}_i(m,j) + \mathit{SEND}_j(m),i) = C
\end{equation*}
\noindent where, $m$ can be a \textit{COMputation message} or an \textit{I am Passive with Credit message}.
\end{inv}

\begin{inv}
\label{inv:strong_termination}
The global strong termination can be declared, in case, there is no PUs in \textit{CRN}. Thus, only a single $\mathit{CR}_i$ contains credit value $C$ if the node is the chief executive node and there is no in-transit message, $m$, in the global channel set, $\mathit{GCS}$.
\begin{equation*}\label{}
\exists i, \forall j: j \in n, i \in j:: hold_i = C \Leftrightarrow i = C_E \wedge \mathit{STATE}(j) = \mathit{PASSIVE} \wedge m \notin GCS	
\end{equation*}
\end{inv}

\begin{inv}
\label{inv:weak_termination}
For the global weak termination, the total credit value $C$ is known to $C_E$. However, $C$ is distributed among $C_E$ and the affected nodes.
\begin{equation*}\label{} out_{C_{E}}[]= \emptyset \wedge in_{C_{E}}[]= \emptyset \wedge (hold_{C_{E}} + C\_PU_{\mathit{affected}_{C_{E}}}[] = C) \end{equation*}
\end{inv}

\subsection{Safety property}
\label{section:Safety property}
The safety property ensures that in no case a node other than $C_E$ announces termination if the computation has indeed terminated. In order to prove the safety property, we consider all the possible cases that may negate the system invariants and violate the safety requirements, as follows:

\begin{enumerate}[noitemsep]
  \item The incorrect recovery from any failure (\textit{e}.\textit{g}., the appearance of PUs, mobility, and crash) may temporarily falsify Invariants~\ref{inv:credit_never_exceeds_C},~\ref{inv:strong_termination}, and ~\ref{inv:weak_termination}. Lemma~\ref{lemma:state_message} and Lemma ~\ref{lemma:recovery_passive_active_nodes} assert that the incorrect recovery from any failure does not violate the safety requirements.
  \item Before reaching the actual termination, the value of $hold_{C_{E}} = C$ or $hold_{C_{E}} > C$, then Invariant~\ref{inv:strong_termination} or Invariant~\ref{inv:weak_termination} are violated. Lemma~\ref{lemma:credit_does_not_exceed1} and Lemma~\ref{lemma:credit_does_not_exceed2} ensure that $C_E$ holds credit $C$ in case of global strong termination and the credit less than $C$ in case of global weak termination.
\end{enumerate}
The proofs of Lemmas~\ref{lemma:state_message}-~\ref{lemma:credit_does_not_exceed2} guarantee the safety requirements of the \textit{T-CRAN} protocol. The following Lemma~\ref{lemma:state_message} and Lemma~\ref{lemma:recovery_passive_active_nodes} prove that the nodes do not violate the safety requirements on their recovery.

\begin{lemma}
\label{lemma:state_message}
The reception of stale messages, $m_{\langle session, *\rangle}$, at the nodes do not increase credit value $C$ forever, which violates the safety requirements of the \textit{T-CRAN} protocol.
\end{lemma}

\begin{proof}
Assume that on recovery,\footnote{We assume that the recovery process takes non-zero time.} $\mathit{CR}_i$ receives stale messages, $m_{\langle session, *\rangle}=m_{\langle x,* \rangle}$, from unreliable channels or other recovered nodes. The reception of $m_{\langle x, * \rangle}$ at $\mathit{CR}_i$ is able to execute the computation and transmission of $m_{\langle x, * \rangle}$, in case $session_i \leq x$. For the contrary, we assume that a node receives a stale message, $m_{\langle x, *\rangle}$, executes the computation and propagates $m_{\langle x, *\rangle}$. We now prove that the reception of stale messages does not violate the safety requirements, as follows:

$\mathit{CR}_i$ can further distribute the computation or surrender credit after completion of its computation among its neighboring nodes, in response to $m_{\langle x, *\rangle}$. The neighboring node $\mathit{CR}_j$ of $\mathit{CR}_i$ may be a recovered node or unaware of the just terminated computation whose $session=x$. Hence, the recipient $\mathit{CR}_j$ can also behave similar to $\mathit{CR}_i$. However, one of the nodes in the network or $C_E$ terminates the flow of $m_{\langle x, *\rangle}$ due to $session_{C_{E}} \neq x$ (Action $A_4$ in Table~\ref{tab:Actions of credit diffusion-aggregation}).

Hence, the system maintains Invariant~\ref{inv:credit_never_exceeds_C}, and once the credit is greater than $C$, it is detected by some nodes; thus stale messages cannot violate the safety requirements of the \textit{T-CRAN} protocol.
\end{proof}

The following assumptions help us to prove Lemma~\ref{lemma:recovery_passive_active_nodes}: we use four different time instants $\alpha, \beta, \gamma$, and $\delta$ such that $\alpha < \beta < \gamma$ (all the other lemmas will also use these time instances) and three nodes $\mathit{CR}_i$, $\mathit{CR}_j$ and $\mathit{CR}_k$ that are neighbors of each other. $\mathit{CR}_i$ initiates the \textit{T-CRAN} protocol at time $\alpha$ among $\mathit{CR}_j$ and $\mathit{CR}_k$ with $session = x$. Under a fault-free scenario, at time $\gamma$, $\mathit{CR}_i$ announces global strong termination. Suppose, $\mathit{CR}_k$ becomes an-affected node at time $\beta$.
\begin{lemma}
\label{lemma:recovery_passive_active_nodes}
On recovery, the initiation of a node in active or passive state does not result in false termination.
\end{lemma}

\begin{proof}
We first mention all the possible situations that may exist at the time of transition of a node from an affected node to a non-affected node or vice versa, which may announce false termination. Afterward, we prove that none of these situations can lead to the violation of the safety requirements in our protocol.

\begin{description}[noitemsep]
  \item[\textsc{Case} 1] $\mathit{CR}_k$ is not able to recover, \textit{i}.\textit{e}., $\mathit{CR}_k$ is an affected node for a very long time.
  \item[\textsc{Case} 2] $\mathit{CR}_k$ recovers, due to availability of another available channel in $LCS_k$ or disappearance of the PU, at time $\delta$, where $\delta$ < $\gamma$.
  \item[\textsc{Case} 3] $\mathit{CR}_k$ recovers, due to availability of another available channel in $LCS_k$ or disappearance of the PU, at time $\delta$, where $\delta$ > $\gamma$.
\end{description}

The \textsc{Case} 1 results in permanent failure of $\mathit{CR}_k$, \textit{i}.\textit{e}., $\mathit{CR}_k$ is a crashed node. Hence, the global strong termination is defied forever, and the protocol announces global weak termination of the computation at $\mathit{CR}_i$ and $\mathit{CR}_j$.

The \textsc{Case} 2 results in the global strong termination, when $session_k=session_{C_{E}} (=session_i)$ at the time of recovery of $\mathit{CR}_k$ in active state. However, passive state of the node is irrelevant here, because passive state of $\mathit{CR}_k$ indicates that $\mathit{CR}_k$ has already surrendered its credit before the transition from a non-affected node to an affection node.

In \textsc{Case} 3, $C_E$ has already declared global weak termination before recovery of $\mathit{CR}_k$. Specifically, \textsc{Case} 1 and \textsc{Case} 3 are almost similar and do not affect $C_E$, because $session_{C_{E}} \neq session_k$. In addition, the recovery of $\mathit{CR}_k$ in active state may cause to propagate messages, $m_{\langle session, *\rangle}=m_{\langle x, * \rangle}$, to $\mathit{CR}_i$ or $\mathit{CR}_j$. However, according to Lemma~\ref{lemma:state_message}, $\mathit{CR}_i$, which is $C_E$, discard $m_{\langle x, *\rangle}$ eventually because $seesion_{C_{E}}\neq x$. (For a better understanding, readers may refer to Figure~\ref{fig:recovery_passive_active_nodes})

Thus, on recovery, the nodes' state do not affect the correct termination, and Invariants~\ref{inv:strong_termination} and~\ref{inv:weak_termination} holds.
\end{proof}

\begin{figure}[t]
\centering
    \begin{minipage}[b]{.23\textwidth}
        \centering
          \includegraphics[]{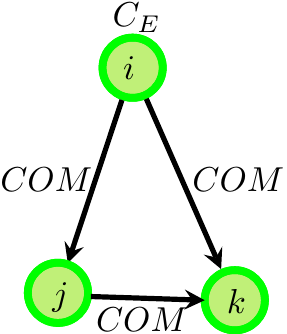}
       \subcaption{At time $\alpha$}
       \label{fig:recovery_passive_active_nodes1}
    \end{minipage}
    \begin{minipage}[b]{.23\textwidth}
        \centering
         \includegraphics[]{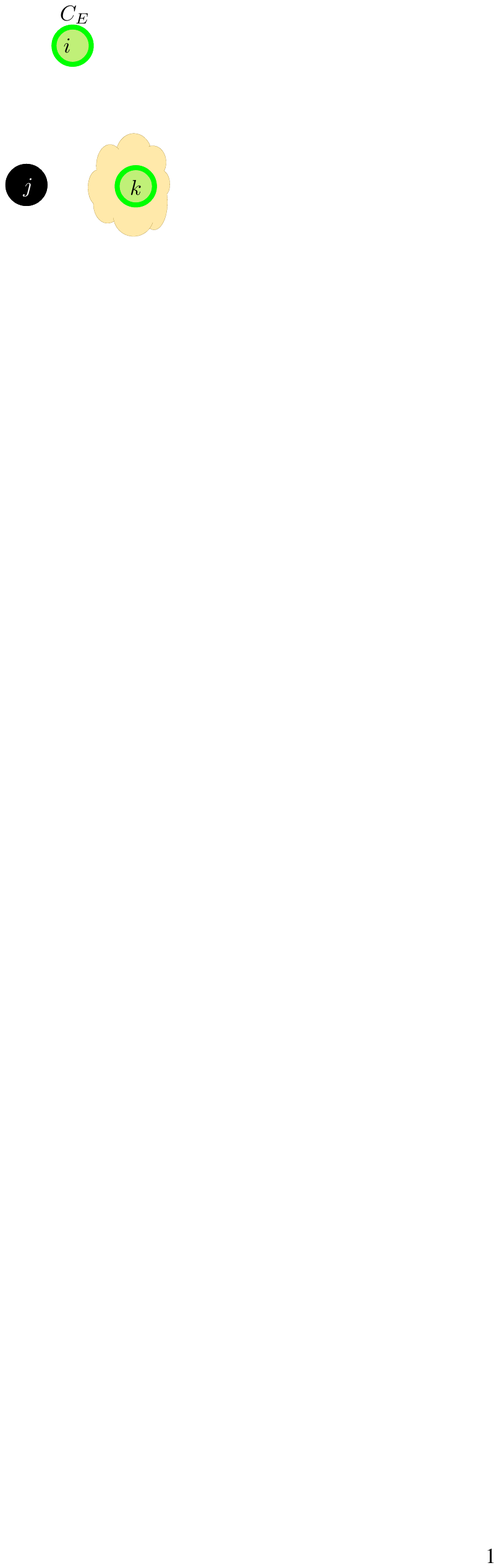}
        \subcaption{At time $\beta$}
        \label{fig:recovery_passive_active_nodes2}
    \end{minipage}
        \begin{minipage}[b]{.23\textwidth}
        \centering
         \includegraphics[]{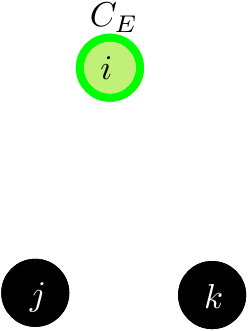}
        \subcaption{At time $\gamma$}
        \label{fig:recovery_passive_active_nodes3}
    \end{minipage}
       \begin{minipage}[b]{.23\textwidth}
        \centering
         \includegraphics[]{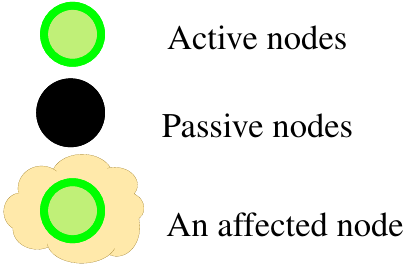}
        \subcaption{Notations}
        \label{fig:recovery_passive_active_nodes4}
    \end{minipage}
\B
\caption{Illustration for the proof of Lemma~\ref{lemma:recovery_passive_active_nodes}.}
\BBB
\label{fig:recovery_passive_active_nodes}
\end{figure}

The credit aggregated at $C_E$ never ever becomes equal to $C$ before the global strong termination is reached. This fact can be justified with the help of Lemma~\ref{lemma:credit_does_not_exceed1} and Lemma~\ref{lemma:credit_does_not_exceed2}, as follows:
\begin{lemma}
\label{lemma:credit_does_not_exceed1}
Under no condition the credit aggregated at $C_E$ equals to $C$ except in the case of global strong termination.
\end{lemma}

\begin{proof}
Suppose, only two processors $\mathit{CR}_i$ and $\mathit{CR}_j$ are executing a computation, and $\mathit{CR}_i$ is the chief executive node. $\mathit{CR}_i$ sends credit $C_j$ to $\mathit{CR}_j$, and again, $\mathit{CR}_j$ sends credit $C_i$ to $\mathit{CR}_i$. Thus, according to Invariant~\ref{inv:credit_never_exceeds_C}, the following equation~\ref{eq:credit_never_exceeds_C} holds true:
\begin{equation}\label{eq:credit_never_exceeds_C}
hold_i + \mathit{SEND}_i(COM(C_j),j) + hold_j + \mathit{SEND}_j(COM(C_i),i) = C
\end{equation}

Suppose at time $\alpha$, $\mathit{CR}_i$ becomes active. Thus, $\mathit{SEND}_i(COM(C_j),j) = 0$. However, at time $\beta$, the following equation~\ref{eq:credit_at_i_from_i_to_j_and_at_j} holds true:
\begin{equation}\label{eq:credit_at_i_from_i_to_j_and_at_j}
hold_i + \mathit{SEND}_i(COM(C_j),j) + hold_j = C
\end{equation}
The above equation~\ref{eq:credit_at_i_from_i_to_j_and_at_j} indicates credit distribution using a \textit{COMputation message} from $\mathit{CR}_i$ to $\mathit{CR}_j$. However, once $\mathit{CR}_j$ receives the \textit{COMputation message}, then $\mathit{SEND}_i(COM(C_j),j) = 0$. Thus, the following equation~\ref{eq:credit_at_i_j} holds ture:
\begin{equation}\label{eq:credit_at_i_j}
hold_i + hold_j = C	
\end{equation}

\begin{figure}[t]
\centering
    \begin{minipage}[b]{.23\textwidth}
        \centering
          \includegraphics[]{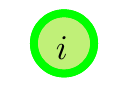}
       \subcaption{At time $\alpha$}
       \label{fig:fig:credit_does_not_exceed1aa}
    \end{minipage}
    \begin{minipage}[b]{.23\textwidth}
        \centering
         \includegraphics[]{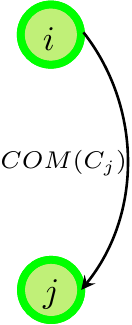}
        \subcaption{At time $\beta$}
        \label{fig:fig:credit_does_not_exceed1bb}
    \end{minipage}
        \begin{minipage}[b]{.23\textwidth}
        \centering
         \includegraphics[]{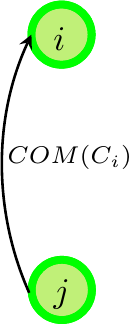}
        \subcaption{}
        \label{fig:fig:credit_does_not_exceed1cc}
    \end{minipage}
       \begin{minipage}[b]{.23\textwidth}
        \centering
         \includegraphics[]{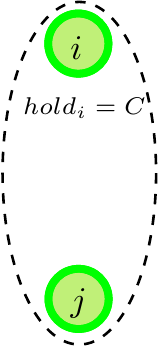}
        \subcaption{At time $\gamma$}
        \label{fig:fig:credit_does_not_exceed1dd}
    \end{minipage}
\B
\caption{Illustration for the proof of Lemma~\ref{lemma:credit_does_not_exceed1}.}
\BBB
\label{fig:credit_does_not_exceed1}
\end{figure}

Assume the contrary, at a later time $\gamma$, $hold_i = C$, and $\mathit{CR}_i$, which is the chief executive node, declares global strong termination, while $\mathit{STATE}(\mathit{CR}_i) = \mathit{STATE}(\mathit{CR}_j) = \mathit{ACTIVE}$. It is possible only when $\mathit{CR}_i$ and $\mathit{CR}_j$ are two processes at an identical node, \textit{i}.\textit{e}., $\mathit{CR}_i = \mathit{CR}_j$. However, the chief executive node, $\mathit{CR}_i$, never declares global strong termination despite $hold_i=C$, unless $\mathit{STATE}(\mathit{CR}_i) = \mathit{PASSIVE}$ (Action $C_2$, Table~\ref{tab:Actions of termination announcement}). Therefore, the protocol never declares termination unless all the nodes are passive, and $hold_{C_{E}}=C$. (For a better understanding, readers may refer to Figure~\ref{fig:credit_does_not_exceed1}.)

The above proof can be generalized for any number of participating nodes. Thus, the protocol always aggregates correct credit in case of the global strong termination.
\end{proof}

The credit aggregated at $C_E$ never becomes (\textit{i}) greater than or equals to $C$, in case of the global weak termination and, (\textit{ii}) greater than $C$, in case of the global strong termination. We call the \textit{wrong credit aggregation} when $hold_{C_{E}}\geq C$ in case of the global weak termination and $hold_{C_{E}}>C$ in case of the global strong termination. Both the above facts are proved by Lemma~\ref{lemma:credit_does_not_exceed2}, as follows:
\begin{lemma}
\label{lemma:credit_does_not_exceed2}
No message transmission ever results in wrong credit aggregation in weak or strong termination declaration.
\end{lemma}

\begin{proof}
The conditions that may lead to wrong credit collection at $C_E$ are the aggregation of an identical credit at more than one node, and the stale messages in the network. However, we have proved that the stale messages are eventually discarded (Lemma~\ref{lemma:state_message}). Hence, we consider the aggregation of an identical credit at more than one node.

We first present a scenario that may lead to multiple times credit surrender of an identical credit at two different nodes. Note that duplicate message reception at a node is handled in the protocol; hence, we not consider it. We present two possible cases that may lead to wrong credit aggregation, and following that we prove by contradiction that these cases never arise in the protocol.

Suppose, an ongoing computation with $session = x$, where $\mathit{CR}_i$ completes its computation and sends an $\mathit{ImPC}(hold_i, b)$ to $\mathit{CR}_j$. In the meantime, suppose $\mathit{CR}_j$ becomes an affected or a failed node (see Failure Model, Section~\ref{section:the_system_model}). Thus, $\mathit{CR}_j$ cannot send an $AcK$ to $\mathit{CR}_i$. Due to non-reception of an $AcK$ from $\mathit{CR}_j$, $\mathit{CR}_i$ sends an $\mathit{ImPC}(hold_i, b)$ to another node, say $\mathit{CR}_k$. However, the recovery of $\mathit{CR}_j$ and the reception of an $\mathit{ImPC}(hold_i, b)$ at $\mathit{CR}_j$ signify that $hold_i$ is surrendered at two different nodes. However, in the protocol, there are only two possible cases of the \textit{imperfect credit surrender}, as follows:
\begin{figure}[t]
\centering
    \begin{minipage}[b]{.24\textwidth}
        \centering
          \includegraphics[]{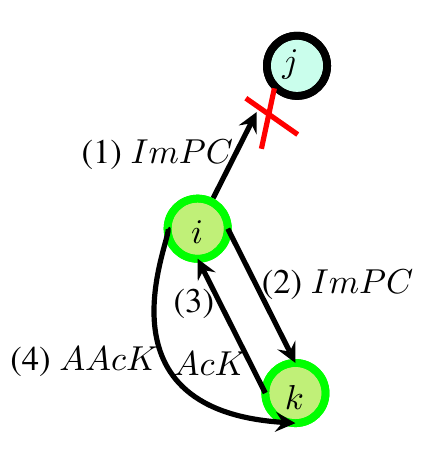}
       \subcaption{}
       \label{fig:credit_does_not_exceed21}
    \end{minipage}
    \begin{minipage}[b]{.24\textwidth}
        \centering
         \includegraphics[]{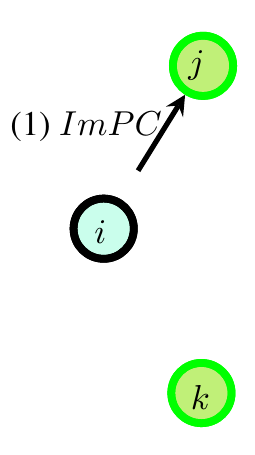}
        \subcaption{}
        \label{fig:credit_does_not_exceed22}
    \end{minipage}
        \begin{minipage}[b]{.24\textwidth}
        \centering
         \includegraphics[]{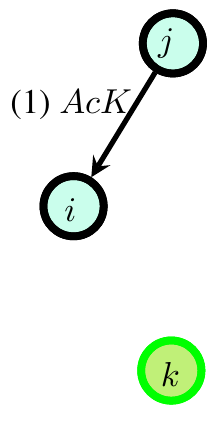}
        \subcaption{}
        \label{fig:credit_does_not_exceed23}
    \end{minipage}
       \begin{minipage}[b]{.24\textwidth}
        \centering
         \includegraphics[]{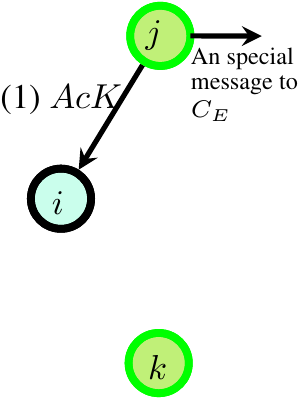}
        \subcaption{}
        \label{fig:credit_does_not_exceed24}
    \end{minipage}

    ~\\

       \begin{minipage}[b]{.99\textwidth}
        \centering
         \includegraphics[]{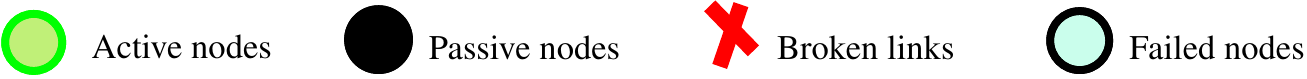}
        \subcaption{Notations}
        \label{fig:ncredit_does_not_exceed25}
   \end{minipage}
\B
\caption{Illustration for the proof of Lemma~\ref{lemma:credit_does_not_exceed2}.}
\BBB
\label{fig:credit_does_not_exceed2}
\end{figure}

\begin{description}[noitemsep]
  \item[\textsc{Case} 1] $\mathit{CR}_j$ is an affected or a failed node after sending an $AcK$ to $\mathit{CR}_i$, and $\mathit{CR}_i$ is already an affected or a failed node after the transmission of an $\mathit{ImPC}(hold_i, b)$.
  \item[\textsc{Case} 2] $\mathit{CR}_i$ sends an $\mathit{ImPC}(hold_i, b)$ to $\mathit{CR}_j$, and due to the absence of an $AcK$ from $\mathit{CR}_j$, $\mathit{CR}_i$ sends an $\mathit{ImPC}(hold_i, b)$ to $\mathit{CR}_k$. After sending $\mathit{ImPC}(hold_i, b)$ messages to two different nodes, $\mathit{CR}_i$ becomes an affected node. Also, $\mathit{CR}_j$ receives the $\mathit{ImPC}(hold_i, b)$.
\end{description}

The reception of an $AcK$ and an $AAcK$ is assumed to be an atomic operation, \textit{i}.\textit{e}., $\mathit{CR}_j$ is not allowed to move to a different location before the reception of an $AAcK$ or a timeout, and also, $\mathit{CR}_j$ receives an $AAcK$ in the presence of PUs. Therefore, \textsc{Case} 1 never holds true.

The \textsc{Case} 2 is essentially an outcome of the \textsc{Case} 1. $\mathit{CR}_j$ receives an $\mathit{ImPC}(hold_i, b)$ and sends an $AcK$ to $\mathit{CR}_i$. As the affected node $\mathit{CR}_i$ cannot receive an $AcK$ from $\mathit{CR}_j$; after a timeout value, $\mathit{CR}_j$ recognizes that $\mathit{CR}_i$ is an affected node. Thus, $\mathit{CR}_j$ sends a special message, $m$, with credit $hold_i$, to $C_E$. Eventually, $C_E$ subtracts $hold_i$ from $hold_{C_{E}}$. On recovery, $\mathit{CR}_i$ again surrender its credit (Lemma~\ref{lemma:recovery_passive_active_nodes}). On the other hand, $\mathit{CR}_i$ receives an $AcK$ from $\mathit{CR}_j$ though it has already surrendered its credit to $\mathit{CR}_k$ earlier; thus, $\mathit{CR}_i$ behaves as an affected node. Therefore, the credit of $\mathit{CR}_i$ remains a constant in the network. (For a better understanding, readers may refer to Figure~\ref{fig:credit_does_not_exceed2}.)

Thus, Invariants~\ref{inv:strong_termination} and~\ref{inv:weak_termination} are preserved, and an imperfect surrender of a credit is infeasible in the \textit{T-CRAN} protocol.
\end{proof}

\subsection{Liveness property}
\label{section:Liveness Property}
In the \textit{T-CRAN} protocol, $C_E$ eventually announces termination in finite time. In order to prove the liveness, we show:
\begin{enumerate}[noitemsep]
  \item The \textit{virtual tree-like structure} does not grow infinitely.
  \item The height of the \textit{virtual tree-like structure} eventually reduces, to (\textit{i}) one, in case of the global strong termination, and (\textit{ii}) two, in case of the global weak termination.
\end{enumerate}

We show that the \textit{virtual tree-like structure} grows and has a finite height. $C_E$, \textit{i}.\textit{e}., the root of the \textit{virtual tree-like structure}, expands the computation among $N$ nodes (using $A_2$ and $A_3$, Table~\ref{tab:Actions of credit diffusion-aggregation}). The distribution of the computation and credit increases the height of the \textit{virtual tree-like structure}, and the maximum height of the \textit{virtual tree-like structure} can be $N$. However, the joining of new nodes in the network during the computation execution may increase the total number of CRs and the maximum height to $N+k, k>0$. In this manner, the \textit{virtual tree-like structure} grows infinitely. However, once the nodes stop to join the network, then the \textit{virtual tree-like structure} does not grow infinitely.

We now show that the height of the \textit{virtual tree-like structure} eventually reduces. The nodes are not allowed to delay the computation for an infinite time. Hence, once all the nodes (except $C_E$), whose heights are identical, complete their computation, they send $\mathit{ImPC}(C,b)$ and $\mathit{ImP}(p)$ messages to their parent node or to any number of neighboring nodes, if they are active (Action $A_4$, Table~\ref{tab:Actions of credit diffusion-aggregation}). Thus, the transmission of $\mathit{ImPC}(C,b)$ and $\mathit{ImP}(p)$ messages by all the nodes (except $C_E$), whose heights are identical (not necessary at an identical time), results in reduction of the height of the \textit{virtual tree-like structure} by at least 1.
\begin{figure}
\centering
    \begin{minipage}[b]{.5\textwidth}
        \centering
          \includegraphics[width=60mm, height=35mm]{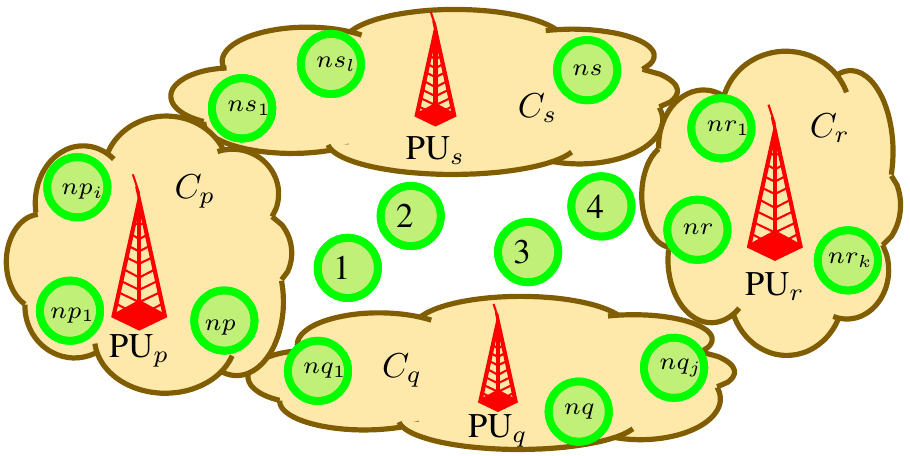}
       \subcaption{}
       \label{fig:Illustrating the impossibility of termination in the cognitive radio networka}
    \end{minipage}
    \begin{minipage}[b]{.49\textwidth}
        \centering
         \includegraphics[width=60mm, height=35mm]{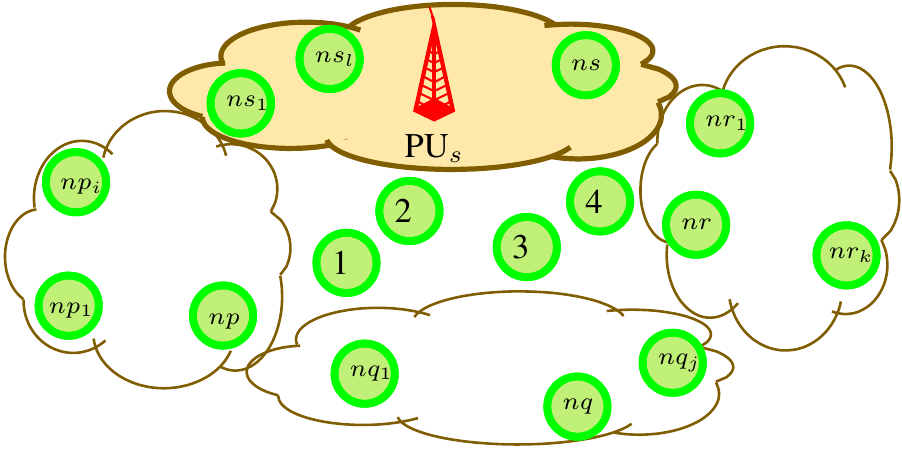}
        \subcaption{}
        \label{fig:Illustrating the impossibility of termination in the cognitive radio networkb}
    \end{minipage}
     \begin{minipage}[b]{.5\textwidth}
        \centering
          \includegraphics[width=60mm, height=35mm]{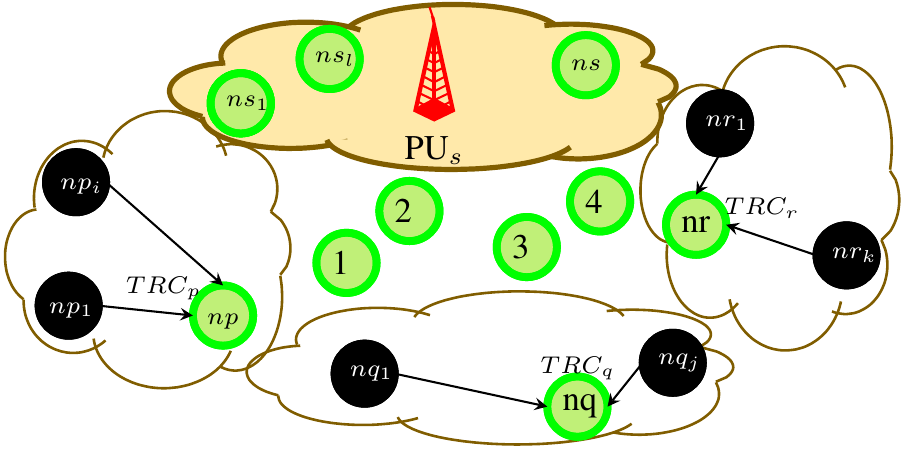}
       \subcaption{}
       \label{fig:Illustrating the impossibility of termination in the cognitive radio networkc}
    \end{minipage}
    \begin{minipage}[b]{.49\textwidth}
        \centering
         \includegraphics[width=60mm, height=35mm]{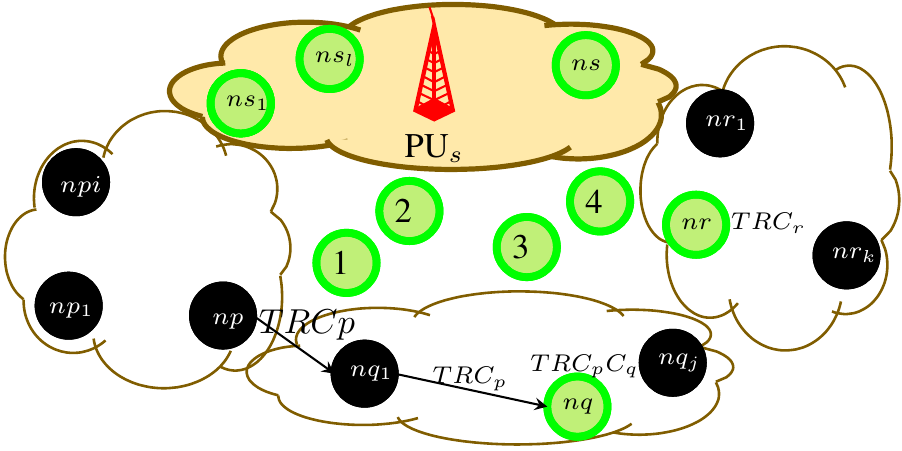}
        \subcaption{}
        \label{fig:Illustrating the impossibility of termination in the cognitive radio networkd}
    \end{minipage}
     \begin{minipage}[b]{.99\textwidth}
        \centering
          \includegraphics[width=60mm, height=35mm]{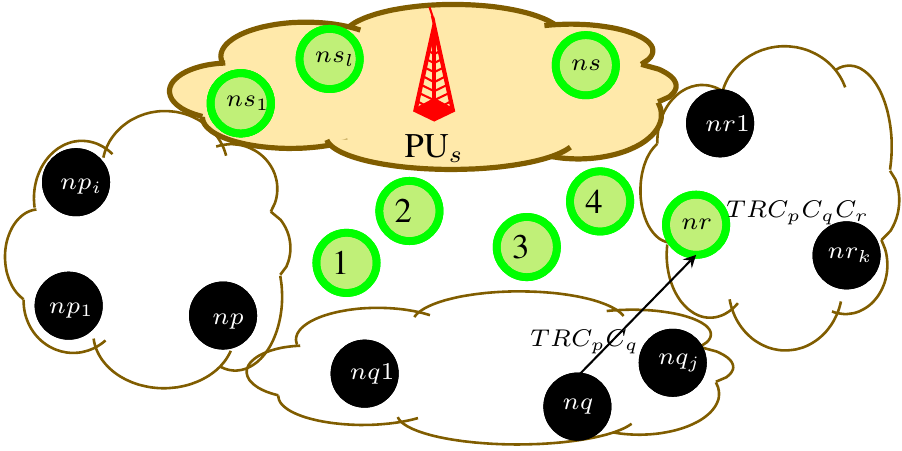}
       \subcaption{}
       \label{fig:Illustrating the impossibility of termination in the cognitive radio networke}
    \end{minipage}
    \begin{minipage}[b]{.99\textwidth}
        \centering
         \includegraphics{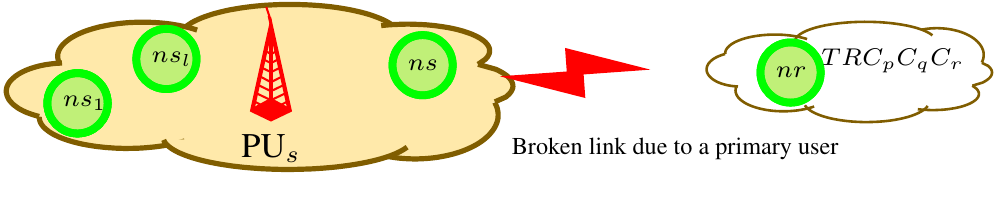}
        \subcaption{}
        \label{fig:Illustrating the impossibility of termination in the cognitive radio networkf}
    \end{minipage}
     \begin{minipage}[b]{.99\textwidth}
        \centering
          \includegraphics{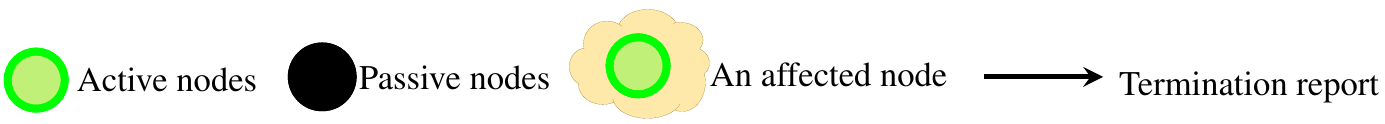}
       \subcaption{Notations}
       \label{fig:notation_Illustrating the impossibility of termination in the cognitive radio network}
    \end{minipage}
\B
\caption{Illustrating the impossibility of termination in cognitive radio networks.}
\B
\label{fig:Illustrating the impossibility of termination in the cognitive radio network}
\end{figure}

\begin{figure}
\centering
          \includegraphics{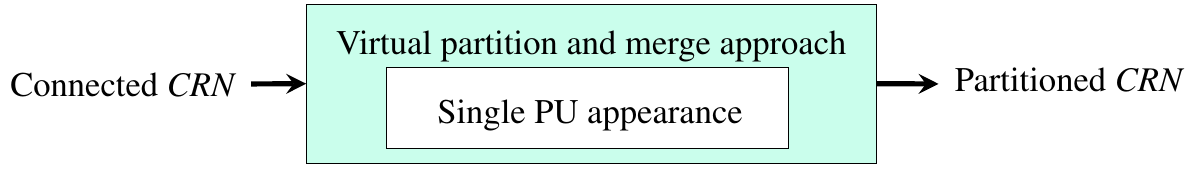}
\B
\caption{The virtual partition and merge approach.}
\BBB
\label{fig:virtual_partition_merge}
\end{figure}

We now show that when the \textit{virtual tree-like structure} has height 1 after credit aggregation, it is a sufficient condition to announce the global strong termination. From the previous facts, it is clear that the height of the \textit{virtual tree-like structure} reduces by at least 1 when all the nodes (except $C_E$), whose heights are identical, send $\mathit{ImPC}(C,b)$ and $\mathit{ImP}(p)$ messages. Hence, when all the nodes of all the height levels (except $C_E$), send $\mathit{ImPC}(C,b)$ and $\mathit{ImP}(p)$ messages that result in the height of the \textit{virtual tree-like structure} to be 1, eventually, and only a single node, $C_E$, holds the complete credit (that is equal to the credit that was distributed at the time of initiation). This fact is enough to show that at the time of the global strong termination the \textit{virtual tree-like structure} has height 1.

We now show that when the \textit{virtual tree-like structure} has height 2 after credit aggregation, it is a sufficient condition to announce the global weak termination. From the previous facts, it is clear that the height of the \textit{virtual tree-like structure} reduces when the non-affected nodes sends $\mathit{ImPC}(C,b)$ and $\mathit{ImP}(p)$ messages. In addition, only a single non-affected node, $C_E$, holds the credit of all the non-affected nodes eventually. Since the affected nodes cannot send $\mathit{ImPC}(C,b)$ and $\mathit{ImP}(p)$ messages, the network is divided into two partitions as: $CRN_P$ and $CRN_N$ (Figure~\ref{fig:The_abstract_view_of_a_primary_user appearance in the cognitive radio network.}). Hence, the \textit{virtual tree-like structure} has height 2, and it is sufficient for the global weak termination.

\subsection{The impossibility of termination}
\label{section:The Impossibility of Termination}
We provide an abstract view to show the impossibility of the global strong termination in the presence of a single primary user. By this abstract view, it will be clear that the appearance of a primary user is difficult to handle than mobility and crash of nodes. (Note that in a purely asynchronous \emph{CRN}, the global strong termination is impossible~\cite{DBLP:journals/jacm/FischerLP85} to detect even if a single primary user exists in the network.)

We consider a \textit{CRN} as a connected communication graph that has nodes $np, np_1, \ldots, np_i$, $nq, nq_1, \ldots, nq_j$, $nr, nr_1, \ldots, nr_k$, $ns, ns_1, \ldots, ns_l$; and no node is assumed to be special. (The node ids are selected in a special way to help readers to understand the abstract view, which will be clear soon.) Also, we assume four PUs, namely PU$_p$, PU$_q$, PU$_r$, and PU$_s$ that affect all the nodes. The \textit{virtual clustering} is performed by considering $np$, $nq$, $nr$, and $ns$ as fixed centers~\cite{erciyes2004cluster} (or cluster heads) that partition the communication graph into four virtual clusters, say $C_p, C_q, C_r$, and $C_s$, see Figure~\ref{fig:Illustrating the impossibility of termination in the cognitive radio networka}.

Now, assume that three PUs, namely PU$_p$, PU$_q$, and PU$_r$, disappear from the network. Thus, the nodes, namely $np, np_1, \ldots, np_i$, $nq, nq_1, \ldots, nq_j$, $nr, nr_1, \ldots, nr_k$, become non-affected nodes, Figure~\ref{fig:Illustrating the impossibility of termination in the cognitive radio networkb}. All the non-affected nodes of each cluster send $\mathit{ImPC}(C,b)$ messages to the respective cluster heads. Hence, each cluster head holds a termination report (or credit) of its cluster, namely $np$ has $TRC_p$, $nq$ has $TRC_q$, and $nr$ has $TRC_r$, Figure~\ref{fig:Illustrating the impossibility of termination in the cognitive radio networkc}.

In order to announce the global strong (or weak) termination in the network, it is required to aggregate all the termination reports (or credits) of each cluster head. Thus, $np$ sends $TRC_p$ to its neighboring virtual cluster head $nq$, and $nq$ aggregates the received credit as: $TRC_pC_q = TRC_p \cup TRC_q$, Figure~\ref{fig:Illustrating the impossibility of termination in the cognitive radio networkd}. This state of the network is equivalent to the \textit{virtual merging} of both the virtual clusters $C_p$ and $C_q$ in one virtual cluster, where $C_q$ is a virtual cluster head. Similarly, $TRC_pC_q$ is aggregated with $TRC_r$ at $nr$ as: $TRC_pC_qC_r = TRC_pC_q \cup TRC_r$, Figure~\ref{fig:Illustrating the impossibility of termination in the cognitive radio networke}.

Since a PU, PU$_s$, persists in the network and $ns, ns_1, \ldots, ns_l$ are affected node, an aggregated termination report (or credit) of the network cannot be generated. The network is now partitioned into two parts, where the first part holds credit $TRC_pC_qC_r$ at $nr$ and the remaining credit is distributed among the affected nodes, $ns, ns_1, \ldots, ns_l$. This state of the network is equivalent to the \textit{virtual partitioning} of the network into two virtual clusters: one virtual cluster, where $C_p$, $C_q$, and $C_r$ are virtually merged and $nr$ is a virtual cluster head, and the another virtual cluster $C_s$ with $ns$ as a virtual cluster head, Figure~\ref{fig:Illustrating the impossibility of termination in the cognitive radio networkf}.

Therefore, it is shown that the existence of a single PU results in two isolated sub-networks, which is sufficient to prevent the protocol to announce the global termination. Without loss of generality, the above \textit{virtual partition-merge} technique (Figure~\ref{fig:virtual_partition_merge}) can be applied to a network of arbitrary size and arbitrary number of PUs.

\textbf{Note}: When the chief executive node, $C_E$, becomes an affected node, neither strong nor weak termination detection is possible.

\end{document}